\documentclass[sigconf]{acmart}

\usepackage{nicefrac}
\usepackage{siunitx}
\usepackage{amsmath}
\usepackage{amsfonts}
\usepackage{array,framed}
\usepackage{booktabs}
\usepackage{multicol}

\usepackage{
  color,
  float,
  epsfig,
  wrapfig,
  graphics,
  graphicx,
  subcaption
}
 \usepackage{setspace}
 \usepackage{latexsym,fancyhdr,url}
\usepackage{enumerate}
\usepackage{algorithm}
\usepackage[noend]{algpseudocode}
\usepackage{graphics}
\usepackage{xparse} 
\usepackage{xspace}
\usepackage{multirow}
\usepackage{csvsimple}
\usepackage{balance}
\usepackage{lineno}

\usepackage{
  tikz,
  pgfplots,
  pgfplotstable
}
\usepackage{hyperref}

\usetikzlibrary{
  shapes.geometric,
  arrows,
  external,
  pgfplots.groupplots,
  matrix
}

\def\cameraReady{} 

\pgfplotsset{compat=1.9}

\usepackage{mathtools,}

\newcommand\StateX{\Statex\hspace{\algorithmicindent}}
\newcommand\StateXX{\StateX\hspace{\algorithmicindent}}
\algrenewcommand\textproc{}

\newcommand{\alglinenoNew}[1]{\newcounter{ALG@line@#1}}
\newcommand{\alglinenoPop}[1]{\setcounter{ALG@line}{\value{ALG@line@#1}}}
\newcommand{\alglinenoPush}[1]{\setcounter{ALG@line@#1}{\value{ALG@line}}}

\newcommand{\sys}{BullShark\xspace}

\DeclareMathAlphabet{\mathcal}{OMS}{cmsy}{m}{n}
\DeclareGraphicsExtensions{%
    .png,.PNG,%
    .pdf,.PDF,%
    .jpg,.mps,.jpeg,.jbig2,.jb2,.JPG,.JPEG,.JBIG2,.JB2}

\usepackage{xparse}
\newcommand{\bnm}{\begin{newmath}}
\newcommand{\enm}{\end{newmath}}

\newcommand{\bea}{\begin{eqnarray*}}%
\newcommand{\eea}{\end{eqnarray*}}%

\newcommand{\bne}{\begin{newequation}}
\newcommand{\ene}{\end{newequation}}

\newcommand{\bal}{\begin{newalign}}
\newcommand{\eal}{\end{newalign}}

\newenvironment{newalign}{\begin{align}%
\setlength{\abovedisplayskip}{4pt}%
\setlength{\belowdisplayskip}{4pt}%
\setlength{\abovedisplayshortskip}{6pt}%
\setlength{\belowdisplayshortskip}{6pt} }{\end{align}}

\newenvironment{newmath}{\begin{displaymath}%
\setlength{\abovedisplayskip}{4pt}%
\setlength{\belowdisplayskip}{4pt}%
\setlength{\abovedisplayshortskip}{6pt}%
\setlength{\belowdisplayshortskip}{6pt} }{\end{displaymath}}

\newenvironment{newequation}{\begin{equation}%
\setlength{\abovedisplayskip}{4pt}%
\setlength{\belowdisplayskip}{4pt}%
\setlength{\abovedisplayshortskip}{6pt}%
\setlength{\belowdisplayshortskip}{6pt} }{\end{equation}}

\newcounter{ctr}

\newcounter{mytable}
\def\mytable{\begin{centering}\refstepcounter{mytable}}
\def\endmytable{\end{centering}}

\newcounter{myfig}
\def\myfig{\begin{centering}\refstepcounter{myfig}}
\def\endmyfig{\end{centering}}

\newlength{\saveparindent}
\newlength{\saveparskip}
\newcommand{\E}{{\rm I\kern-.3em E}}

\renewcommand{\eqref}[1]{\mbox{Equation~(\ref{#1})}}

\def \part {part}

\renewcommand{\paragraph}[1]{\vspace*{6pt}\noindent\textbf{#1}\;}

\newtheorem{observation}{Observation}
\newtheorem{claim}{Claim}

\newtheorem{property}{Property}

\def \blackslug{\hbox{\hskip 1pt \vrule width 4pt height 8pt
    depth 1.5pt \hskip 1pt}}
\def \qed{\quad\blackslug\lower 8.5pt\null\par}

\newcounter{mynote}[section]

\newcommand\ignore[1]{}

\newcounter{rcnote}[section]

\newcounter{mrnote}[section]

\newcounter{fknote}[section]

\newcounter{anote}[section]

\DeclareMathSymbol{\mlq}{\mathord}{operators}{``}
\DeclareMathSymbol{\mrq}{\mathord}{operators}{`'}

\newcommand{\rhf}[2]{R_{f, \gamma}}

\DeclareDocumentCommand{\edist}{o o}{
  \ensuremath{
    \IfNoValueTF{#1}{{d}}{{\sf d}(#1,#2)}
  }
}

\newcommand{\olrk}[1]{\ifx\nursymbol#1\else\!\!\mskip4.5mu plus 0.5mu\left(\mskip0.5mu plus0.5mu #1\mskip1.5mu plus0.5mu \right)\fi}

\NewDocumentCommand{\indseq}{ O{1} O{r} }{{#1}\ldots {#2}}

\algdef{SE}[Upon]{Upon}{EndUpon}[1]{\textbf{upon}\ #1\ \algorithmicdo}{\algorithmicend\ \textbf{}}%
\algtext*{EndUpon}

\newcommand{\sasha}[1]{[\textbf{\textcolor{blue}{Sasha:}} { \textcolor{blue}{#1}]}}

\usepackage{cleveref}
\begin{document}

\def\thetitle{Bullshark: DAG BFT Protocols Made Practical}
\title{\thetitle}

\ifdefined\cameraReady
  \author{Alexander Spiegelman}
  \email{sasha.spiegelman@gmail.com}
  \affiliation{Aptos}

  \author{Neil Giridharan}
  \email{giridhn@berkeley.edu}
  \affiliation{University of California, Berkeley}

  \author{Alberto Sonnino}
  \email{alberto@sonnino.com}
  \affiliation{Mysten Labs}

  \author{Lefteris Kokoris-Kogias}
  \email{Lefteris2k@gmail.com}
  \affiliation{IST Austria}
\else
  \author{}
\fi

\date{}
\renewcommand{\shortauthors}{Spiegelman, et al.}

\begin{abstract}
 
 We present \sys, the first directed acyclic graph (DAG) based asynchronous Byzantine Atomic Broadcast protocol that is optimized for the common synchronous case. Like previous DAG-based BFT protocols~\cite{keidar2021all, danezis2021narwhal}, \sys requires no extra communication to achieve consensus on top of building the DAG.
 That is, parties can totally order the vertices of the DAG by interpreting their local view of the DAG edges.
 Unlike other asynchronous DAG-based protocols, \sys provides a practical low latency fast-path that exploits synchronous periods and deprecates the need for notoriously complex view-change and view-synchronization mechanisms.
 \sys achieves this while maintaining all the desired properties of its predecessor DAG-Rider~\cite{keidar2021all}. Namely, it has optimal amortized communication complexity, it provides fairness and asynchronous liveness, and safety is guaranteed even under a quantum adversary.

 In order to show the practicality and simplicity of our approach, we also introduce a standalone partially synchronous version of \sys,
 which we evaluate against the state of the art.
 The implemented protocol is embarrassingly simple (200 LOC on top of an existing DAG-based mempool implementation~\cite{danezis2021narwhal}). It is highly efficient, achieving for example, 125,000 transactions per second with a 2 seconds latency for a deployment of 50 parties. In the same setting, the state of the art pays a steep 50\% latency increase as it optimizes for asynchrony.
 
\end{abstract}

\maketitle

\section{Introduction}
\label{sec:intro}

Ordering transactions in a distributed Byzantine environment via a consensus mechanism has become one of the most timely research areas in recent years due to the blooming Blockchain use-case.
A recent line of
work~\cite{gkagol2019aleph,baird2016swirlds,keidar2021all,danezis2021narwhal,schett2021embedding, yang2019prism} proposed an
elegant way to separate between the dissemination of transactions and the
logic required to safely order them. 
The idea is simple. 
To propose transactions, parties send them in a way that forms a
casual order among them.
That is, messages contain blocks of transactions as well as
references to previously received messages, which together form a
\emph{directed acyclic graph (DAG)}.
Interestingly, the structure of the DAG encodes information that allow
parties to totally order the DAG by locally interpreting their view of it
without sending any extra messages.
That is, once we build the DAG, implementing consensus on top of it requires \emph{zero-overhead} of communication. 

The pioneering work of Hashgraph~\cite{baird2016swirlds} constructed an unstructured DAG,
where each message refers to two previous ones, and 
used hashes of messages as local coin flips to
totally order the DAG in asynchronous settings.
Aleph~\cite{gkagol2019aleph} later introduced a structured round-based DAG and encoded
a shared randomness in each round via a threshold signature scheme to achieve constant latency in expectation.
The state of the art is DAG-Rider~\cite{keidar2021all}, which is built on previous ideas.
Every round in its DAG has at most $n$ vertices (one for each party), each of which contains a block of transactions as well as references (edges) to at least $2f+1$ vertices in the previous round. Blocks are disseminated via reliable broadcast~\cite{bracha1987asynchronous} to avoid equivocation, and an honest party advances to the next round once it reliably delivers $2f+1$ vertices in the current round.
Remarkably, by using the DAG to abstract away the communication layer, the entire edges interpretation logic of DAG-Rider to totally order the DAG spans less than 30 lines of pseudocode.

DAG-Rider is an asynchronous Byzantine atomic broadcast
(BAB), which achieves optimal amortized communication complexity ($O(n)$ per transaction), post quantum
safety, and some notion of fairness (called Validity) 
that guarantees that every transaction proposed by an honest party is eventually delivered (ordered).
To achieve optimal amortized communication DAG-Rider combines batching techniques with an efficient asynchronous verifiable information dispersal protocol~\cite{cachin2005dispersal} for the reliable broadcast building block. 
The protocol is post quantum safe because it does not rely on primitives that a quantum computer can break for the safety properties.
That is, a quantum adversary can prevent the protocol progress, but it cannot violate safety guarantees.  

Although DAG-based protocols have a solid theoretical
foundation, they have multiple gaps before being realistically deployable in practice.
First, they all optimize for the worst case asynchronous network
assumptions and do not take advantage of synchronous periods, resulting to higher latency than existing consensus protocols~\cite{hotstuff, tendermint} in the good case.
Second, they have some impractical assumptions such as needing unbounded
memory in order to preserve fairness. The only existing solution to this comes from Tusk~\cite{danezis2021narwhal}, 
which uses a garbage collection mechanism but does not allow for quantifiable fairness even during periods of synchrony. 

On the other hand, existing partially synchronous consensus protocols are designed as a monolith, where the leader of the protocol has to propose blocks of transactions in the critical path, resulting in performance bottlenecks and relatively low throughput as shown by Narwhal~\cite{danezis2021narwhal}.

To the best of our knowledge, this paper is the first to optimize the DAG-based BFT approach to the partially synchronous communication setting.
First, we propose \sys, which preserves all the theoretical properties of DAG-Rider (including asynchronous worst case liveness), and in
addition, introduces a fast path that exploits common-case synchronous network conditions.
That is, \sys is the first BAB protocol with optimal amortized
communication complexity ($O(n)$ per transaction) and post quantum safety that is optimized for the common case. \sys needs only 2 round-trips between commits during synchrony (thus a $75\%$ improvement compared to DAG-Rider), and maintains a 6 round-trip expected latency in asynchronous executions (matching DAG-Rider).
In addition, \sys is built on top of Narwhal and thus inherits all of its practical benefits (e.g., decoupling data dissemination from the DAG construction and having an efficient reliable broadcast implementation).

Second, based on \sys's fast path, we present an eventually synchronous variant of \sys, which is the first partially synchronous consensus protocol that is completely embedded into a DAG.
The protocol is fundamentally different from previous partially synchronous protocols since it is symmetric, and \emph{does not require a view-change or view synchronization mechanisms} after a faulty leader.
The resulting protocol is embarrassingly simple and extremely efficient, achieving 125k TPS and 2 second latency with 50 honest parties.       
As a final contribution, \sys overcomes an existing practical limitation of DAG-based protocols of having to choose between fairness and garbage collection.
\sys garbage collects vertices belonging to old DAG rounds, and also provides fairness during synchronous periods. 
As an evidence to its practicality, the partially synchronous version of \sys has already been productionized by Mysten Labs and is currently being integrated by Aptos. 





In summary, this paper makes the following contributions:
\begin{itemize}
    \item We propose \sys, the first slow-path/fast-path DAG-based consensus protocol that achieves significantly lower latency than prior work. \sys takes 2 rounds in the good case and 6 rounds in expectation (matching DAG-Rider) in asynchrony.
    \item We simplify \sys to work only in partial synchrony. This version of \sys results in a significantly simpler partially synchronous consensus protocol than prior work (extra 200LOC vs 4000LOC of Hotstuff over a DAG~\cite{danezis2021narwhal}). \sys additionally performs significantly better under faults making it the most performant and resilient partially synchronous protocol to date.
    \item We show how to build a practical DAG-based system that allows for garbage collection and provides timely fairness after GST, answering an open question of prior work~\cite{keidar2021all,danezis2021narwhal}.
\end{itemize}

\section{Technical challenges.}



In order to design and implement \sys we had to solve a number of theoretical and
practical challenges.

\paragraph{Theoretical challenges.}
The approach in current DAG-based protocols is to advance rounds
as soon as enough messages in the current round are received ($2f+1$
for Aleph and DAG-Rider).
This works perfectly for asynchronous consensus, but unfortunately cannot guarantee deterministic liveness during synchronous periods~\cite{fischer1985impossibility}, as required by the eventually synchronous variant of \sys.
This is because the adversary can, for example, reorder messages (within the synchrony bound) to make sure parties advance rounds before getting messages from the predefined leaders. Note that this is inherent to any deterministic protocol. We considered and evaluated two alternatives (see Appendix~\ref{app:logicaldag}) and decided to embed timeouts into the DAG construction as it provided better performance. 
In a nutshell, if the first $2f+1$ messages in a round do not contain one from the leader, then parties wait for a timeout or a message from the leader before advancing to the next round. 

A further challenge is to take advantage of a common-case
synchronous network without sacrificing latency in the asynchronous worst case.
To this end, \sys introduces two types of votes - \emph{steady-state} for the predefined leader and \emph{fallback} for the random one.
Similarly to DAG-Rider~\cite{keidar2021all}, \sys rounds are grouped in \emph{waves}, each of which consists of 4 rounds.
Intuitively, each wave encodes the consensus logic.
The first round of a wave has two potential leaders - a predefined steady-state leader and a leader that is chosen in retrospect by the randomness produced in the fourth round of the wave.
To reduce latency in synchronous periods, the third round of a wave also has a predefined leader. It takes two rounds to commit a steady-state leader.
Based on their voting type, the vertices in the second round can potentially vote for the steady-state leader in the first round 
and vertices in the fourth round can potentially vote for the fallback leader in round one or the steady-state leader in round three.
Importantly, the same vertex cannot vote for both the fallback and steady-state leaders in the same wave.
A vertex's voting type is determined by whether or not its source (the party that broadcasted it) committed a leader in the previous wave.
This information is encoded in the DAG and since the DAG is built on top of a reliable broadcast abstraction, even Byzantine parties cannot lie about their voting type.

A nice property of \sys is that it does not require a view change or view synchronization mechanisms to overcome faulty or slow leaders. Instead of a view change, \sys uses the information encoded in the DAG to maintain safety.
Since all parties agree on the causal histories of vertices they have in the DAG, after a leader is  committed each party locally ``rides'' the DAG (wave by wave) backwards to see which leader-vertices could have been committed by other parties.
Synchronizing views is not required because (as we show in our proofs) the DAG construction already provides it. 
If the first leader in a wave after GST is honest, then all parties advance to the third round of the wave roughly at the same time. 

\paragraph{Practical challenges.}
Finally, to evaluate \sys we had to resolve some practical challenges. 
First, all previous theoretical solutions require unbounded memory
to hold the entire DAG, and second, the reliable broadcast primitive we
use to clearly describe \sys (used in DAG-Rider and Aleph as well)
is inefficient in the common-case.
Fortunately, Narwhal~\cite{danezis2021narwhal} implemented a scalable
DAG and dealt exactly with these problems.
We started from Narwhal's open source codebase and adopt
their approach to decouple data from metadata to implement an efficient broadcast.
Unfortunately, the Narwhal garbage collection mechanism
directly conflicts with \sys's mechanism to provide fairness.
In fact, providing meaningful fairness for all honest parties seems to be
impossible with bounded memory implementations in asynchronous networks
since every message can be delayed to after the relevant prefix of
the DAG is garbage collected.
To deal with this issue we relax our fairness requirement.
That is, our bounded memory implementation of \sys guarantees \textit{timely fairness}
only during synchronous periods. This means that after GST all messages by honest
parties make it into the DAG in finite time and before the garbage
collection. For all the other messages (before GST) we use Tusk's approach of retransmission, where guarantees can only be made for an unbounded execution.

\section{Preliminaries}
\label{sec:prelimnaries}

\subsection{Model}

We consider a peer to peer message passing model with a set of $n$
parties $\Pi = \{ p_1, \ldots, p_n \}$, and a \emph{dynamic} adversary
that can corrupt up to  $f < n/3$ of them during an execution.
We say that corrupted parties are \emph{Byzantine}  and all other
parties are \emph{honest}. 
Byzantine parties may act arbitrarily, while honest ones follow the
protocol.
We assume that the adversary is computationally bounded.

For the description of the protocol we assume that links between 
honest parties are reliable.
That is, all messages among honest parties eventually arrive~\footnote{We address this issues from a practical
point of view in our implementation.}.
Moreover, for simplicity, we assume that recipients can verify the
senders identities.
We assume a known $\Delta$ and say that an execution of a
protocol is \emph{eventually synchronous} if there is a \emph{global
stabilization time (GST)} after which all messages sent among honest
parties are delivered within $\Delta$ time.
An execution is \emph{synchronous} if GST occurs at time 0, and
\emph{asynchronous} if GST never occurs.

For the protocol analysis we are interested in the practical
performance as well as theoretical complexity during synchronous and
asynchronous periods, or alternatively, before and after the GST.
To this end, we define consider the following scenarios: 
 
\begin{itemize}
  
  \item \emph{Worst case condition}: asynchronous execution and $f$
  byzantine parties
  
  \item \emph{Common case condition}: synchronous executions
  with no failures~\footnote{Same analysis apply to eventually
  synchronous failure-free executions after GST.}
  
\end{itemize}

\subsection{Building blocks}

Similarly to DAG-Rider, we use the following known building blocks
for our modular protocol presentation:

\paragraph{Reliable broadcast}
Each party $p_k$ can broadcast messages by calling
$\textit{r\_bcast}_k(m,r)$, where $m$ is a message and $r \in \mathbb{N}$
is a round number.
Every party $p_i$ has an output $\textit{r\_deliver}_i(m,r,p_k)$,
where $m$ is a message, $r$ is a round number, and $p_k$ is the party
that called the corresponding $\textit{r\_bcast}_k(m,r)$.
The reliable broadcast abstraction guarantees the following properties:
\begin{description}
    \item[Agreement] If an honest party $p_i$ outputs
    $\textit{r\_deliver}_i(m,r,p_k)$, then every other honest party
    $p_j$ eventually outputs \\ $\textit{r\_deliver}_j(m,r,p_k)$.
    \item[Integrity] For each round $r \in \mathbb{N}$ and party $p_k
    \in \Pi$, an honest party $p_i$ outputs
    $\textit{r\_deliver}_i(m,r,p_k)$ at most once regardless of~$m$.
    \item[Validity] If an honest party $p_k$ calls
    $\textit{r\_bcast}_k(m,r)$, then every honest party $p_i$
    eventually outputs $\textit{r\_deliver}_i(m,r,p_k)$.
\end{description}

\paragraph{Global perfect coin} 
An instance $w$, $w \in \mathbb{N}$, of the coin is invoked by party
$p_i \in \Pi$ by calling $\textit{choose\_leader}_i(w)$.
This call returns a party $p_j \in \Pi$, which is the chosen leader
for instance $w$.
Let $X_w$ be the random variable that represents the probability that
the coin returns party $p_j$ as the return value of the call
$\textit{choose\_leader}_i(w)$.
The global perfect coin has the following guarantees:
\begin{description}
    \item[Agreement] If two honest parties $p_i,p_j$ call
    $\textit{choose\_leader}_i(w)$ and $\textit{choose\_leader}_j(w)$ with respective return values $p_1$ and $p_2$, then $p_1=p_2$.
    \item [Termination] If at least $f+1$ honest parties call $\textit{choose\_leader}(w)$, then every $\textit{choose\_leader}(w)$ call eventually returns.
    \item[Unpredictability] As long as less than $f+1$ honest parties call $\textit{choose\_leader}(w)$, the return value is indistinguishable from a random value except with negligible probability $\epsilon$.
    Namely, the probability $pr$ that the adversary can guess the returned
    party $p_j$ of the call $\textit{choose\_leader}(w)$ is $pr \leq
    \Pr [X_w = p_j] + \epsilon$. 
    \item[Fairness] The coin is fair, i.e., $\forall w \in \mathbb{N}, \forall p_j \in \Pi \colon \Pr[X_w = p_j] = 1/n$.
\end{description}

Implementation examples that use PKI and a threshold signature
scheme~\cite{libert2016born, boneh2001short, shoup2000practical} can be found
in~\cite{cachin2005Constantinople, loss2018combining}.
See DAG-Rider for more details on how a coin implementation can be
integrated into the DAG construction.
It is important to note that the above mentioned implementations satisfy Agreement, Termination, and Fairness with information theoretical guarantees. 
That is, the assumption of a computationally bounded adversary is required only for the unpredictability property.
As we later prove, the unpredictability property is only required for Liveness.
Therefore, since similarly to DAG-Rider generating randomness is the only place where cryptography is used, the Safety properties of \sys are post-quantum secure.

\subsection{Problem Definition}
\label{sec:problemDef}
Following DAG-Rider~\cite{keidar2021all}, our result focuses on the
\emph{Byzantine Atomic Broadcast (BAB)} problem.
%
%
To avoid confusion with the events of the underlying reliable broadcast
abstraction, the broadcast and deliver events of BAB are
$\textit{a\_bcast}(m,r)$ and $\textit{a\_deliver}(m,r,p_k)$,
respectively, where $m$ is a message, $r \in \mathbb{N}$
is a sequence number, and $p_k \in \Pi$ is a party. 
The purpose of the sequence numbers is to distinguish between messages
broadcast by the same party.
We assume that each party broadcasts
infinitely many messages with consecutive sequence numbers.

\begin{definition} [Byzantine Atomic Broadcast] Each honest
party $p_i \in \Pi$ can call $\textit{a\_bcast}_i(m,r)$ and output
$\textit{a\_deliver}_i(m,r,p_k)$, $p_k \in \Pi$.
A Byzantine Atomic Broadcast protocol satisfies reliable broadcast
(agreement, integrity, and validity) as well as:
\begin{description}
    \item[Total order] If an honest party $p_i$ outputs
    $a\_deliver_i(m,r,p_k)$ before $a\_deliver_i(m',r',p_k')$, then no honest 
    party $p_j$ outputs $a\_deliver_j(m',r',p_k')$ before
    $a\_deliver_j(m,r,p_k)$.
\end{description}

\end{definition}

Note that the above definition is agnostic to the network assumptions.
However, in asynchronous executions, due to the FLP
result~\cite{fischer1985impossibility}, BAB cannot be solved
deterministically and therefore we relax the validity property to hold
with probability $1$ in this case.
Moreover, the validity property cannot be satisfied in asynchronous
executions with bounded memory implementation.
Therefore, as we discuss more in Section~\ref{sec:GC}, for the practical
version of this problem, we require validity to be satisfied only after
GST in eventually synchronous executions.


Note that the BAB abstraction captures the core consensus logic in
permissioned blockchain systems as it provides a mechanism to propose blocks of
transactions and totally order them.
Moreover, similar to Hyperledger~\cite{androulaki2018hyperledger}, it 
supports a separation between the total order mechanism and transaction
execution.
Transaction validation can therefore be done as part of the
execution~\cite{androulaki2018hyperledger} before applying it to the SMR.

\section{DAG Construction}

In this section we describe our DAG construction and explain how it is
different from the one in DAG-Rider~\cite{keidar2021all}.
In a nutshell, DAG-Rider is a fully asynchronous atomic broadcast
protocol and thus rounds in its DAG advance in network speed as soon
as $2f+1$ nodes from the current round are delivered. 
Here, we are interested in a protocol that deterministically achieves better latency in synchronous periods. 
Therefore, introducing timeouts into the system is unavoidable~\cite{fischer1985impossibility}.
We considered and evaluated two alternatives (see Appendix~\ref{app:logicaldag} for more details) and decided to integrate timeouts into the DAG construction.
It is important to note that despite the timeouts, our DAG still advances
in network speed when the leader is honest. 

We present the background, structures, and basic utilities we borrow from
DAG-Rider in Section~\ref{sub:dagBackground}.
We describe our DAG construction in Section~\ref{sub:ourDAG}.

\subsection{Background}
\label{sub:dagBackground}

\begin{algorithm*}[t]
    \caption{Data structures and basic utilities for party $p_i$}
    \label{alg:dataStructures}
    \begin{algorithmic}[1]
    \footnotesize
    
        \Statex \textbf{Local variables:}
        \StateX struct $\textit{vertex } v$: 
        \Comment{The struct of a vertex in the DAG}
        \StateXX $v.\textit{round}$ - the round of $v$ in the DAG
        \StateXX $v.\textit{source}$ - the party that broadcast $v$
        \StateXX $v.\textit{block}$ - a block of transactions
        \StateXX $v.\textit{strongEdges}$ - a set of vertices in
        $v.\textit{round}-1$ that represent \emph{strong} edges 
        \StateXX $v.\textit{weakEdges}$ - a set of vertices in rounds $<
        v.\textit{round}-1$ that represent \emph{weak} edges 
        \StateX $DAG_i[]$ - An array of sets of vertices, initially:
            \StateXX $DAG_i[0] \gets$ predefined hardcoded set of $2f+1$
            ``genesis'' vertices 
            \StateXX  $\forall j \geq 1 \colon DAG_i[j] \gets \{\}$ 
        \StateX $\textit{blocksToPropose}$ - A queue, initially empty,
        $p_i$ enqueues valid blocks of transactions from clients

        \vspace{0.5em}
        \Procedure{\textit{path}}{$v,u$} \Comment{Check if exists a path consisting of strong and weak edges in the DAG}
        \State \Return exists a sequence of $k \in \mathbb{N}$,
        vertices $v_1,v_2,\ldots,v_k$  s.t.\
        \StateXX $v_1 = v $, $v_k = u$, and $\forall i \in [2..k]
        \colon v_i \in \bigcup_{r \geq 1}
        DAG_i[r] \wedge (v_i \in v_{i-1}.\textit{weakEdges} \cup
        v_{i-1}.\textit{strongEdges})$
        \EndProcedure
        
        \vspace{0.5em}
        \Procedure{\textit{strong\_path}}{$v,u$} \Comment{Check if exists a path consisting of only strong edges in the DAG}
        \State \Return exists a sequence of $k \in \mathbb{N}$,
        vertices $v_1,v_2,\ldots,v_k$  s.t.\
        \StateXX $v_1 = v $, $v_k = u$, and $\forall i \in [2..k]
        \colon v_i \in \bigcup_{r \geq 1}
        DAG_i[r] \wedge v_i \in v_{i-1}.\textit{strongEdges}$
        \EndProcedure
        
        \begin{multicols}{2}
        \Procedure{\textit{create\_new\_vertex}}{round} \label{alg:DAG:createNewVertex}
        \State \textbf{wait until}
        $\neg$\textit{blocksToPropose}.\text{empty}()
        \State $v.round \gets round$
        \State $v.source \gets p_i$  
        \State $v.\textit{block} \gets
        \textit{blocksToPropose}.\text{dequeue}()$
        \State $v.\textit{strongEdges}
        \gets DAG[\textit{round} - 1]$ 
        \State $\textit{set\_weak\_edges}(v,\textit{round})$ \State \textbf{return} $v$
        \EndProcedure
        
        \vspace{0.5em}
        \Procedure{\textit{set\_weak\_edges}}{$v, \textit{round}$} 
        \Comment{Add edges to orphan vertices} 
        \label{alg:DAG:getWeakEdges} 
        \State $v.\textit{weakEdges} \gets \{\}$
        \For{$r=\textit{round}-2$ down to 1}
        \For{\textbf{every} $u \in DAG_i[r]$ s.t. $\neg \textit{path}(v,u) $}
        \State $v.\textit{weakEdges} \gets v.\textit{weakEdges} \cup \{u \}$
        \EndFor
        \EndFor
        \EndProcedure

        \vspace{0.5em}
        \Procedure{\textit{get\_fallback\_vertex\_leader}}{$w$}
            \State $p \gets \textit{choose\_leader}_i(w)$
            \State \Return $get\_vertex(p,4w-3)$ 
        \EndProcedure
        
           \vspace{0.5em}
        \Procedure{\textit{get\_first\_steady\_vertex\_leader}}{$w$}
            \State  $p \gets \textit{get\_first\_predefined\_leader}(w)$
            \State \Return $get\_vertex(p,4w-3)$ 
        \EndProcedure

         \vspace{0.5em}
        \Procedure{\textit{get\_second\_steady\_vertex\_leader}}{$w$}
            \State  $p \gets \textit{get\_second\_predefined\_leader}(w)$
            \State \Return $get\_vertex(p,4w-1)$ 
        \EndProcedure
        
        \vspace{0.5em}
        \Procedure{\textit{get\_vertex}}{p,r}
        
            \If{$\exists v\in DAG[r]$ s.t.\ $v.source = p$}
                \State \Return $v$ 
            \EndIf 
                \State \Return $\bot$ 
        
        \EndProcedure
        
        \end{multicols}

         \alglinenoNew{counter}
         \alglinenoPush{counter}

    \end{algorithmic}
\end{algorithm*}

We use a DAG to abstract the communication
layer among parties and enable the establishment of common knowledge.
Each vertex in the DAG represents a message disseminated via reliable broadcast from a single party, containing, among other data,
references to previously broadcasted vertices.
Those references are the edges of the DAG.
Each honest party maintains a local copy of the DAG, and
different honest parties might observe different views of it (depending on the order in which
they deliver the vertices). Nevertheless, reliable broadcast prevents equivocation and
guarantees that all honest parties eventually deliver the same
messages, hence their views of the DAG eventually converge.

The DAG data types and and basic utilities are
specified in \Cref{alg:dataStructures}.
For each party $p_i$, we denote $p_i$'s local view of the
DAG as $DAG_i$, which is represented by an array of sets of vertices $DAG_i[]$.
Vertexes are created via the $create\_new\_vertex(r)$ procedure.
Each vertex in the DAG is associated with a
unique round number $r$ and the party who generated and reliably broadcasted it (the source).
In addition, each vertex $v$ contains a block of transactions that were
previously $a\_bcast$ by the BAB protocol that is implemented on top of the DAG and two sets of
outgoing edges.
The set \emph{strong edges} contains at least $2f+1$ references to
vertexes associated with round $r-1$ and the set \emph{weak
edges} contains up to $f$ references to vertices in rounds $< r-1$
such that otherwise there is no path from $v$ to them.
As explained in the next sections, strong edges are used for
Safety and weak edges make sure we eventually include all
vertices in the total order, to satisfy BAB's validity property.

The entry $DAG_i[r]$ for $r \in \mathbb{N}$ stores a set
of vertices associated with round $r$ that $p_i$ previously
delivered.
By the reliable broadcast, each party can broadcast at most $1$
vertex in each round and thus $|DAG_i[r]| \leq n$.

The procedures $path(v,u)$ and $strong\_path(v,u)$ get two vertexes and
check if there is a path from $v$ to $u$. 
The difference between them is that $path(v,u)$ considers all edges
while $strong\_path(v,u)$ only considers the strong ones.

The procedure $get\_fallback\_vertex\_leader$ gets a
wave number, computes the randomly elected leader of the wave and then returns the vertex that the elected leader broadcast in the first round of the wave, if it
is included in the DAG.
Otherwise, returns $\bot$.
Similarly, the procedures $get\_first\_steady\_vertex\_leader$ and $get\_second\_steady\_vertex\_leader$ return the vertices broadcast by the first and second predefined leaders of the
wave, respectively.
We assume a predefined and known to all parties mapping waves to steady-state leaders.

\subsection{Our DAG protocol}
\label{sub:ourDAG}

A detailed pseudocode is given in \Cref{alg:DAG}.
Each party $p_i$ maintains three local variables: \emph{round} stores the last round in which $p_i$ broadcast a vertex, \emph{buffer} stores vertices that where reliably delivered but not yet added to the DAG, and \emph{wait} is an Boolean that indicate whether the timeout for the current round has already expired.
Each party $p_i$ is constantly trying to advance rounds and calling the high-level BAB protocol to totally order all the vertices in its DAG.
When $p_i$ advances its round, it broadcast its vertex for this round and start a timeout.

Our DAG protocol is triggered by one of two events:
a vertex delivery (via reliable broadcast) or a timeout expiration.
Once a party $p_i$ delivers a vertex it first checks if the vertex is legal, i.e., (1) the source and round must match the reliable broadcast instance to prevent equivocation, and (2) the vertex must has at least $2f+1$ strong edges.
Then, $p_i$ checks if the vertex is ready to be added to the DAG by calling $try\_add\_to\_DAG$.
The idea is to make sure that the causal history of a vertex is always available in the DAG. 
Therefore, a vertex is added to the DAG only if all the vertices it includes as references are already delivered.
If this is not yet the case, the vertex is added to a \emph{buffer} for a later retry.
Once a vertex $v$ is added to the DAG, the high-level BAB protocol is invoked, via the $try\_ordering(v)$ interface, to check if more vertices can now be totally ordered.

We next describe the conditions for advancing rounds.
Note that since DAG-Rider only cares about the asynchronous case, rounds are advanced as soon as $2f+1$ vertices in the current round are delivered.
We, in contrast, optimize for the common case conditions and thus have to make sure that parties do not advance rounds too fast.
Otherwise, the adversary can prevent honest parties from committing steady-state leaders since it controls which $2f+1$ vertexes parties deliver first even after GST.
Therefore, we keep the DAG-Rider necessary condition (in $try\_advance\_round$) but extend it to make sure that honest steady-state leaders are committed in network speed after GST.

We distinguish between slow and up-to-date parties.
As mentioned in the introduction, \sys does not require an external view-synchronization mechanism for slow parties. Instead, once $p_i$ delivers $2f+1$ vertices in a round $r > \emph{round}$, $p_i$ jumps forward to round $r$, broadcasts a vertex in round $r$, and starts a new timeout. 

For the up-to-date parties we need to be more careful.
As we explained more in the next section, each wave has a steady-state leader in the first round and a steady-state leader in the third one.
Intuitively, the vertices of these leaders are interpreted as "proposals" and the vertices in immediately following rounds with strong edges to the leaders' vertices are interpreted as "votes".
In addition, each party can vote for the steady-state leaders in a wave only if its voting type is steady-state for this wave.
To make sure all honest parties get a chance to vote for steady state leaders, an up-to-date honest party $p_i$ will try to advance (via $try\_advance\_round$) to the second and forth rounds of a wave only if (1) the timeout for this round expired or (2) $p_i$ delivered a vertex from the wave predefined first and second steady-state leader, respectively.
Similarly, we need to make sure the adversary cannot prevent honest parties from collecting enough votes to commit an honest leader after GST.
Therefore, before trying to advance (via $try\_advance\_round$) to the third round a wave or the first round of the next wave, $p_i$ waits for either the timeout expiration or to deliver $2f+1$ vertices in the current round with steady-state voting type and strong edges to the first and second steady-leader, respectively.
In Section~\ref{sub:ESproof}, we prove that after GST timeouts never expire for honest leaders and the DAG advances in network speed.

\begin{algorithm}
    \caption{DAG construction, protocol for process $p_i$}
    \begin{algorithmic}[1]
    \footnotesize
    
    \alglinenoPop{counter}
    
        \Statex \textbf{Local variables:}
        
            \StateX $\emph{round} \gets 1$; 
            $buffer \gets \{\}$;
            $\emph{wait} \gets true$

        \Statex      

        \Upon{$\textit{r\_deliver}_i(v,r,p)$}
        
        \If{$v.source = p \wedge v.round = r \wedge \left|
        v.\textit{strongEdges} \right| \geq 2f+1$}

            \If{$\neg try\_add\_to\_dag(v)$}
            
                \State $\textit{buffer} \gets \textit{buffer} \cup \left\{
                v \right\}$ 
            
            \Else 
            
                \For{$v \in buffer: v.round \leq r$}
                
                    \State \hspace*{3mm} $try\_add\_to\_dag(v)$
                
                \EndFor
                
           \EndIf
           
           \If{$r = round$}
                
                \State $\emph{w} \gets \lceil r/4 \rceil $
                \Comment{steady state wave number}
                
                \If{$ r~mod~4 = 1 \wedge (\neg 
                \emph{wait} \vee \exists v \in DAG[r]: v.source = 
                \hspace*{13mm} get\_first\_steady\_vertex\_leader(\emph{w}))$}
                \label{line:advance1}    
                    \State \hspace*{3mm} $try\_advance\_round()$
                    
                \EndIf

                \If{$ r~mod~4 = 3 \wedge (\neg 
                \emph{wait} \vee \exists v \in DAG[r]: v.source = 
                \hspace*{13mm} get\_second\_steady\_vertex\_leader(\emph{w}))$}
                \label{line:advance3} 
                    
                    \State \hspace*{3mm} $try\_advance\_round()$
                    
                \EndIf
                
                \If{$ r~mod~4 =  0 \wedge (\neg \emph{wait}
                \vee \exists U \subseteq DAG[r]: |U| = 
                \hspace*{13mm} 2f+1 \text{ and }
                \forall u \in U, u.source \in \emph{steadyVoters}[
                w]) \wedge 
                \hspace*{15mm} strong\_path(u, get\_second\_steady\_leader(\emph{w}))$}
                \label{line:advance4} 
                    \State \hspace*{3mm} $try\_advance\_round()$
                
                \EndIf
                
                \If{$ r~mod~4 =  2 \wedge (\neg \emph{wait}
                \vee \exists U \subseteq DAG[r]: |U| = 
                \hspace*{13mm} 2f+1 \text{ and }
                \forall u \in U, u.source \in \emph{steadyVoters}[
                w]) \wedge 
                \hspace*{15mm} strong\_path(u, get\_first\_steady\_leader(\emph{w}))$}
                \label{line:advance2} 
                
                    \State \hspace*{3mm} $try\_advance\_round()$
                
                \EndIf
            \EndIf
            
        \EndIf
        \EndUpon
        
        \Statex
        
        \Upon{timeout}
        
            \State $\emph{wait} \gets false$
            \State $try\_advance\_round()$ 
        
        \EndUpon
        
        \Statex
        
        \Procedure{try\_add\_to\_dag}{$v$}
        
            \If{$\forall v' \in v.\emph{strongEdges} \cup
            v.\emph{weakEdges}: v' \in \bigcup\limits_{k \geq 1}
            DAG[k]$} 
            
                \State $DAG[v.round] \gets DAG[v.round] \cup \{v\}$
                \If{$|DAG[v.round]| \geq 2f+1 \wedge v.round > \emph{round}$}
                    \State $\emph{round} \gets v.round$; 
                    \emph{start timer};
                    \label{line:sync}
                    $\emph{wait} \gets true$
                    \label{line:time2}
                    \Comment{Synchronize waves}
                    \State $broadcast\_vertex(v.round)$
                
                \EndIf
                \State $\textit{buffer} \gets \textit{buffer} \setminus \{ v\}$
                \State $try\_ordering(v)$
                \State \Return true 

            \EndIf
            \State \Return false
        
        \EndProcedure
        
        \Statex
        
        \Procedure{try\_advance\_round()}{}
        
            \If {$\left| DAG[round] \right| \geq 2f + 1$}
                \State $\emph{round} \gets \emph{round}+1$; \emph{start timer};
                $\emph{wait} \gets true$
                \label{line:time1}
                \State $broadcast\_vertex(round)$

               
            \EndIf
            
        \EndProcedure
        
        \StateX
        
        \Procedure{broadcast\_vertex(r)}{}
        
                \State $v \gets \textit{create\_new\_vertex}(r)$ 
                \State $try\_add\_to\_dag(v)$
                \State $\textit{r\_bcast}_i(v,r)$

               
            
        \EndProcedure
        
        \alglinenoPush{counter}

\end{algorithmic}
\label{alg:DAG}
\end{algorithm}

\section{The \sys Protocol}
\label{sec:protocol}

In this section we present a detailed description of \sys.
Similarly to DAG-Rider~\cite{keidar2021all}, the ordering logic of \sys requires no communication on top of building the DAG.
Instead, each party observes its local copy of the DAG and totally order its vertices by interpreting the edges as "votes".
In order to optimize for the common case conditions while guaranteeing liveness under worst case asynchronous conditions, \sys has two types of leaders: \emph{steady-state} and \emph{fallback}.
The main challenge in designing \sys is the interplay between
them as we need to make sure parties cannot vote for both types at the same round.
Illustration of \sys can be found in Figure~\ref{fig:bullsharkexample}.
We divide the protocol description into two parts.
In Section~\ref{sub:modes} we
describe the commit rule of each leader, and in Section~\ref{sub:ordering} we
explain how parties totally order leaders' causal histories.
In Section~\ref{sec:ES} we preset an eventually synchronous version of \sys and in Section~\ref{sec:GC} we discuss the details of our garbage collection mechanism.
For space limitations, we provide formal proofs for both versions on \sys in Appendix~\ref{app:proofs}.

\subsection{Voting Types}
\label{sub:modes}

\begin{figure}
    \centering
    \includegraphics[width=0.48\textwidth]{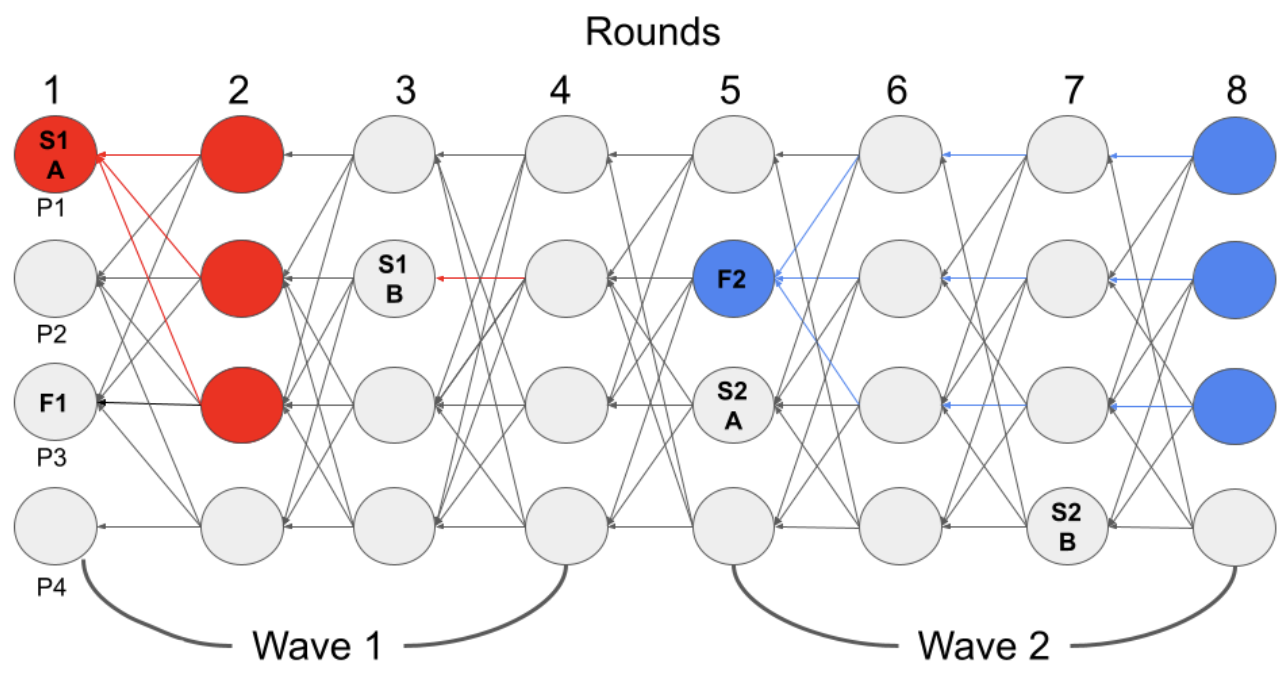}
    \caption{\textbf{Illustration of the DAG at party $\mathbf{P1}$. The columns represent the round 
    numbers and the rows are all the vertices from a particular party ($\mathbf{P1, P2, P3, P4}$ top to bottom). $\mathbf{S1A}$ denotes the first steady state leader of wave $\mathbf{1}$ (in round $\mathbf{1}$), 
    and $\mathbf{S1B}$ denotes the second steady state leader of wave $\mathbf{1}$ (in round $\mathbf{3}$). $\mathbf{F1}$ denotes the 
    fallback leader of wave $\mathbf{1}$ (in round $\mathbf{1}$). All parties start off with a steady state vote
    type in wave $\mathbf{1}$. In round $\mathbf{2}$, $\mathbf{P1}$ observes $\mathbf{3}$ $\mathbf{(2f+1)}$ steady state votes for $\mathbf{S1A}$ 
    (denoted in red), so $\mathbf{P1}$ commits $\mathbf{S1A}$. In round $\mathbf{4}$, $\mathbf{P1}$ only observes $\mathbf{1}$ vote for the
    second steady state leader $\mathbf{S1B}$, so $\mathbf{P1}$ does not commit $\mathbf{S1B}$. Since $\mathbf{P1}$ does not commit
    the second steady state leader, it has a fallback vote type in wave $\mathbf{2}$. From the DAG in
    round $\mathbf{5}$, $\mathbf{P1}$ also observes that $\mathbf{P2, P3, P4}$ did not commit $\mathbf{S1B}$, so all
    parties have a fallback vote type in wave $\mathbf{2}$. Thus $\mathbf{S2A}$ and $\mathbf{S2B}$ (the first and second steady state leaders in wave $\mathbf{2}$ respectively) cannot be committed since all vote types are fallback. In round $\mathbf{8}$, $\mathbf{P1}$ observes $\mathbf{3}$ 
    ($\mathbf{2f+1}$) fallback votes for the fallback leader $\mathbf{F2}$ (denoted in blue), so $\mathbf{P1}$ commits 
    $\mathbf{F2}$. Once $\mathbf{P1}$ commits $\mathbf{F2}$, it checks to see whether any previous leader it did not commit, could have 
    been committed. In round $\mathbf{4}$, $\mathbf{P1}$ only observes $\mathbf{1}$ steady state vote for $\mathbf{S1B}$ (less 
    than $\mathbf{f+1}$), so it does not commit $\mathbf{S1B}$ since if it would have been committed by some party then $\mathbf{P1}$ would observed at least $f+1$ votes.}}
    
    \label{fig:bullsharkexample}
\end{figure}
\begin{algorithm}
\caption{\sys part 1: $p_i$'s alg. to update parties vote type}
\begin{algorithmic}[1]
 \footnotesize

\alglinenoPop{counter}

        \Statex \textbf{Local variables:}
        
        \StateX $\emph{steadyVoters}[1] \gets \Pi$; $\emph{fallbackVoters}[1] \gets \{\}$
        \StateX \textbf{For every} $j > 1$,
        $\emph{steadyVoters}[j],\emph{fallbackVoters}[j] \gets \{\}$
    
       \Statex
       
       \Upon{$\textit{a\_bcast}_i(b,r)$} 
        
        \State $\textit{blocksToPropose}.\text{enqueue}(b)$
        
        \EndUpon
        
    \StateX

    \Procedure{try\_ordering}{$v$}
        \State $w \gets \lceil v.round / 4 \rceil$
        \State $\emph{votes} \gets v.\emph{strongEdges}$
     \If{\emph{v.round} mod 4 = 1 }
        \Comment{first round of a wave}
        
         \State \emph{determine\_party\_vote\_type(v.source, votes, w)}
         
    \ElsIf{\emph{v.round} mod 4 = 3}
        \State \emph{try\_steady\_commit}(\emph{votes}, \emph{get\_first\_steady\_vertex\_leader(w)}, $w$)
    
    \EndIf
        
    \EndProcedure
    
       \StateX 
       
    \Procedure{determine\_party\_vote\_type}{$p,votes, w$}
        
        \State $v_s \gets \emph{get\_second\_steady\_vertex\_leader(w-1)}$
        \State $v_f \gets \emph{get\_fallback\_vertex\_leader(w-1)}$
        
        \If{\emph{try\_steady\_commit}(\emph{votes}, $v_s$, $w-1$) $\vee$ \emph{try\_fallback\_commit}(\emph{votes, $v_f$, $w-1$})}
        
            \State $\emph{steadyVoters}[w] \gets \emph{steadyVoters}[w] \cup \{p\}$
            
        \Else
        
            \State $\emph{fallbackVoters}[w] \gets \emph{fallbackVoters}[w] \cup \{p\}$
        
        \EndIf
        
        
            
        
%
%
        

        
        
                
        
        
    
    \EndProcedure
    
    \Statex

    \Procedure{try\_steady\_commit}{\emph{votes}, $v, w$}
        
        \If{$|\{ v' \in \emph{votes}: 
        v'.source \in \emph{staedyVoters[w]} \wedge  
        \hspace*{5mm} strong\_path(v',v)\}| \geq 2f+1$}
            \State $commit\_leader(v)$
            \State \Return true
        \EndIf
        \State \Return false
    \EndProcedure

    \Statex
    
    \Procedure{try\_fallback\_commit}{\emph{votes}, $v, w$}
    
        \If{$|\{ v' \in \emph{votes}: 
        v'.source \in \emph{fallbackVoters[w]} \wedge  
        \hspace*{5mm} strong\_path(v',v)\}| \geq 2f+1$}
            \State $commit\_leader(v)$
            \State \Return true
        \EndIf
        \State \Return false
    \EndProcedure
    
    \alglinenoPush{counter}

\end{algorithmic}
\label{alg:modes}
\end{algorithm}

Similarly to DAG-Rider, to interpret the DAG, each party $p_i$ divides its local view of the DAG, $DAG_i$, into waves of 4 rounds each.
Unlike DAG-Rider, which has one potential leader in every wave, \sys has three.
One \emph{fallback} leader in the first round of each wave, which is elected retrospectively via the randomness produced in the forth round of the wave (as in DAG-Rider), and two predefined \emph{steady-state} leaders in the first and third rounds of each wave.
In the common case, during synchronous periods, both steady-state leaders are committed in each wave, meaning that it takes two rounds on the DAG to commit a leader.
During asynchronous periods, each fallback leader is committed with probability of at least $2/3$. Meaning that during asynchrony, a fallback leader is committed every $6$ rounds in expectation and \sys has liveness with probability $1$.

A nice property of the common case execution of \sys is that it does not require external view-change and view-synchronization mechanisms.
When switching from asynchrony to synchrony, the first two rounds of each wave make sure that if the first leader is honest then all honest parties start the third round roughly at the same time.
View-change is not required because the DAG encodes all the information needed for safety. In particular, parties can see what information other parties had when they interpreted the DAG, and decide accordingly.

The pseudocode appears in~\Cref{alg:modes}.
The procedure try\_ordering is called every time a new vertex is added to the DAG.
Since \sys has two types of leaders in each wave, we need to ensure that fallback and steady-state leaders are never committed in the same wave.
To this end, parties cannot vote for both types of leaders in the same wave.
That is, every party is assigned with a voting type in every wave that is either fallback or steady-state.
When a party $p_i$ interprets its local copy of the DAG it keeps track of other parties voting types in
\emph{steadyVoters[w]} and \emph{fallbackVoters[w]}, where $w$ is a wave number.

Intuitively, a party is in \emph{steadyVoters[w]} if it has committed either the second steady-state or the fallback leader in wave $w-1$.
Specifically, party $p_i$ determines $p_j$'s voting type in wave $w$ when it delivers $p_j$'s vertex $v$ in the first round of wave $w$, which triggers the call to the determine\_party\_vote\_type procedure.
If the causal history of $v$ has enough information to commit one of these leaders, then $p_i$ determines $p_j$'s voting type as steady-state, otherwise, as fallback.
By the properties of reliable broadcast, all parties see the same causal history of vertex $v$, and thus agree on $p_j$'s voting type in round $w$ (even Byzantine parties cannot lie about their voting type).

To commit a leader in wave $w-1$ based on a vertex $v$ in the first round of a wave $w$, $p_i$ considers the set of vertices pointed by $v$'s strong edges as potential "votes".
Note that these vertices belong to wave $w-1$ and each of them has a voting type that was already previously determined by $p_i$.
To commit the fallback leader of wave $w-1$, at least of $2f+1$ out of the potential votes must have strong paths to the leader and a fallback voting type.
Similarly, to commit the second steady-state leader of wave $w-1$, at least $2f+1$ out of the potential votes must to have strong paths to the leader and steady-state voting type.
Committing the first steady-state leader of a wave is similar but in this case the strong edges of a vertex in the third round of the wave are considered as potential votes.
Note that since even a Byzantine party cannot lie about its voting type, quorum intersection guarantees that leaders with different types cannot be committed in the same wave.
This is the reason we ask for $2f+1$ strong paths unlike Tusk where $f+1$ strong paths are sufficient for safety.
As we describe next, when a leader $v$ is committed then the procedure commit\_leader is called to totally order $v$'s causal history.

\subsection{Ordering The DAG}
\label{sub:ordering}

So far we described the wave commit rules and how parties use them to determine other parties voting types.
Next we describe how we totally order the DAG.
The pseudocode appears in Algorithm~\ref{alg:commit}.
Once a party $p_i$ commits a (steady-state or fallback) leader vertex $v$ it calls $commit\_leader(v)$.
To totally order the causal history of $v$, $p_i$ first tries to commit previous leaders for which the commit rule in its local copy of the DAG was not satisfied.
To do this, $p_i$ traverses back the rounds of its DAG until the last round in which it committed a leader and check whether it is possible that other honest parties committed leaders in these rounds based on their local copy of the DAG.
If $p_i$ encounters such a leader, it orders it before $v$.
Note that this part is much trickier than in DAG-Rider since \sys has three potential leaders in every wave.


\begin{algorithm}
\caption{\sys part 2: the commit alg. for party $p_i$}
\begin{algorithmic}[1]
 \footnotesize

        \alglinenoPop{counter}

        \Statex \textbf{Local variables:}
      
        \StateX $ \emph{committedRound} \gets 0$
        \StateX $\emph{deliveredVertices} \gets \{\}$
        \StateX $\emph{leaderStack} \gets $ initialize empty stack
    
    \Statex
    
    \Procedure{commit\_leader}{$v$}
    
        \State $\emph{leaderStack}.push(v)$
        \State $r \gets \emph{v.round} -2 $
        \Comment{There is a potential leader to commit every two rounds}
        \While{$r > \emph{committedRound}$}
        
            \State $w \gets \lceil r/4 \rceil $
            \State $\emph{ssPotentialVotes} \gets
            \{v' \in DAG_i[r+1] ~|~ strong\_path(v,v') \}$

            \If{$ r~mod~4 == 1$}
            \Comment{two potential leaders in this round}
            
               \State $v_s \gets get\_first\_steady\_vertex\_leader(w)$ 
               
               \State $v_f \gets get\_fallback\_vertex\_leader(w)$ 
               
               \State $\emph{ssVotes} \gets \{v' \in \emph{ssPotentialVotes}: v'.source \in$ \Statex \hspace{21mm} $\emph{steadyVoters}[w] \wedge  strong\_path(v',v_s) \}$ 
               
               \If{$v.round = r+2$}
               
                    \State $\emph{fbVotes} \gets \{\}$
                    \Comment{fallback leader could not be committed since \State \hspace{9mm} there at least $2f+1$ steady-state vote types in this wave }
               \Else
                    \State  $\emph{fbPotentialVotes} \gets
                    \{v' \in DAG_i[r+3] ~|~ strong\_path(v,v') \}$
                    \State $\emph{fbVotes} \gets \{v' \in \emph{fbPotentialVotes}: v'.source \in$ \Statex \hspace{26mm} $\emph{fallbackVoters}[w] \wedge  strong\_path(v',v_f) \}$

               \EndIf

            \Else
            \Comment{$ r~mod~4 == 3$}
            
                \State $v_s \gets get\_second\_steady\_vertex\_leader(w)$
                
                \State $\emph{ssVotes} \gets \{v' \in \emph{ssPotentialVotes}: v'.source \in$ \Statex \hspace{21mm} $\emph{steadyVoters}[w] \wedge  strong\_path(v',v_s) \}$ 
               
               \State $v_f \gets \bot$; $\emph{fbVotes} \gets \{\}$

            \EndIf
            
                    \If{$|\emph{ssVotes}| \geq f+1 \wedge |\emph{fbVotes}| <f+1$}
                    \label{line:indirectfirst}
            \State $leadersStack.push(v_s)$
            \label{line:stacksteady}
            \State $v \gets v_s$ 
            
        \EndIf
        
        \If{$|\emph{ssVotes}| < f+1 \wedge |\emph{fbVotes}| \geq f+1$}
         \label{line:indirectfallback}
        
            \State $leadersStack.push(v_f)$
            \label{line:stackfallback}
            \State $v \gets v_f$
            \label{line:indirectlast}
            
        \EndIf
        
        \State $r \gets r -2$
       \EndWhile

        \State $\emph{committedRound} \gets v.round$ 
        \State $order\_vertices()$

            
            
        
        
        

            
        
        
            


    \EndProcedure
    
    \Statex
    
    \Procedure{$order\_vertices()$}{}
        \While{$\neg \textit{leadersStack}.\text{isEmpty}()$} 
        \State $v \gets \textit{leadersStack}.\text{pop}()$ \label{alg:SMR:stackPop}
          \State \textit{verticesToDeliver} $\gets \{v' \in \bigcup_{r > 0}
        DAG_i[r] \mid path(v,v') \wedge v' \not\in $
        \Statex \hspace{30mm} $\emph{deliveredVertices}\}$
          \For{$\textbf{every} ~v' \in \textit{verticesToDeliver}$ in some deterministic
          order}
              \State \textbf{output}
              $\textit{a\_deliver}_i(v'.\textit{block},v'.\textit{round},
              v'.\textit{source})$
              \label{alg:SMR:decide}
              \State $\textit{deliveredVertices} \gets \textit{deliveredVertices} \cup \{v'\}$
          \EndFor
        \EndWhile
        \EndProcedure  
    
    \alglinenoPush{counter}

\end{algorithmic}
\label{alg:commit}
\end{algorithm}

By quorum intersection and the non-equivocation property of
the DAG, if some party commits either a fallback or a steady-state leader by seeing $2f+1$ votes, then all other parties see at least $f+1$ of these votes.
Moreover, since a party cannot vote for both types of leaders in the same wave, if $p_i$ sees $f+1$ votes for the fallback (steady-state) leader, then no party could have committed the steady-state (fallback) leader since in this case there are at most $2f$ votes with steady-state (fallback) type.

To make sure $p_i$ orders the leaders that precedes $v$ consistently with the other parties, we need to make sure that parties consider the same potential votes when deciding whether to order one of them.
To this end, to decide whether to order a steady-state leader $v'$, $p_i$ sets the potential votes to be all the vertices in round $v'.round +1$ in its DAG such that there is a strong path between the last leader $p_i$ previously ordered and $v'$.
For a fallback leader $v'$, the potential votes are set in a similar way but round $v' + 3$ is used instead of $v'+1$ to be consistent with the commit rule.

After computing the potential votes, $p_i$
checks if one of the leaders in the round it is currently traversing could be committed by other honest parties.
First, $p_i$ checks the potential votes type and the existence of strong paths to the leaders to determines the sets of votes for the steady-state and fallback leaders.
Note that the set of votes for the fallback leader is empty in rounds without a fallback leader or if a steady-state leader was already committed in this wave.
Then, $p_i$ checks if one of the leaders $u$ in the round has at least $f+1$ votes while the other has at most $f$.
If this is the case $p_i$ orders $u$ by pushing it to the leader's stack \emph{leaderStack} and continues its traversal to the next rounds to check if there are leaders to order before $u$.
Otherwise, $p_i$ skips the leaders of the current round as it is guaranteed that none if them could have been committed.


As we prove in Appendix~\ref{app:proofs}, all honest parties order the same leaders and in the same order.
All that is left is to apply some deterministic rule to order their
causal histories one by one.
Therefore, after committing a leader $v$ (and finishing ordering all leaders that proceeds $v$ for which the commit rule was not satisfied), party $p_i$ calls
$order\_vertex()$.
This function goes over the ordered leaders one by one, and for each of them delivers, by some deterministic order, all
the blocks in the vertices in it causal history (strong and weak edges) that have not yet
been delivered.

\section{Eventually synchronous Bullshark}
\label{sec:ES}

In this section we present an eventually synchronous version of the Bullshark protocol.
This protocol is embarrassingly simple, and as we demonstrate in Section~\ref{sec:evaluation}, very efficient.
To the best of our knowledge, this is the first eventually synchronous BFT protocol that does not require view-change or view-synchronization mechanism.
The presentation here is based on the terminology of Section~\ref{sec:protocol}. An intuitive illustration can be found in Appendix~\ref{app:PSB} and an extended description in~\cite{DAGmeetsBFT}.

In a nutshell, there are no fallback leaders in the eventually synchronous version of \sys. Instead, parties keep trying to commit the steady-state leaders.
The pseudocode, which overwrites the $try\_ordering$ procedure, appears in Algorithm~\ref{alg:ESBullshark} (Note that some procedures from previous Algorithms are called). 
In section~\ref{sub:ESproof} we give a formal proof of Safety and Liveness. In a nutshell, the safety proof has a similar proof structure as \sys with fallback, and for liveness we show that after GST two consecutive honest predefined leaders guarantee that the second leader will be committed by all honest parties.
In particular, we show that if the first leader of wave $w$ is honest, then all honest parties advance to the third round of $w$ roughly at the same time. Moreover, if the second leader is honest than all honest parties will wait for the second leader before advancing to the fourth round, and thus all honest will see at least $2f+1$ votes for the second leader in $w$ and commit it.

\begin{algorithm}
\caption{Eventually synchronous \sys: alg. for party
$p_i$.}
\begin{algorithmic}[1]
 \footnotesize

\alglinenoPop{counter}

        \Statex \textbf{Local variables:}

        \StateX $\emph{committedRound} \gets 0$
        \StateX $\emph{leaderStack} \gets $ initialize empty stack
    
       \Statex

        \Procedure{try\_ordering}{$v$}
        \State $w \gets \lceil v.round / 4 \rceil$
        \State $\emph{votes} \gets v.\emph{strongEdges}$
     \If{\emph{v.round} mod 4 = 1 }
        \Comment{try committing second leader of  prev wave}
        
         \State \emph{try\_commit}(\emph{votes}, \emph{get\_second\_steady\_vertex\_leader(w-1)})
         
    \ElsIf{\emph{v.round} mod 4 = 3}
    \Comment{try committing first leader of this wave}
        \State \emph{try\_commit}(\emph{votes}, \emph{get\_first\_steady\_vertex\_leader(w)})
    
    \EndIf
        
    \EndProcedure
       
     \StateX
    
    \Procedure{try\_commit}{$votes, v$}
        
        \If{$|\{ v' \in votes: strong\_path(v',v)\}|
        \geq f+1$}
        
            \State $commit\_leader(v)$
        
        \EndIf
    \EndProcedure
    
    \Statex

    \Procedure{commit\_leader}{$v$}
    
        \State $leadersStack.push(v)$
        \State $r \gets v.round -2$
        \While{r > \emph{committedRound}}
        \State $w \gets \lceil v.round / 4 \rceil$

            \If{$ r~mod~4 == 1$}
            
               \State $v_s \gets get\_first\_steady\_vertex\_leader(w)$
            
            \Else
            \Comment{$ r~mod~4 == 3$}
            
                \State $v_s \gets get\_second\_steady\_vertex\_leader(w)$

            \EndIf
            
            \If{$strong\_path(v,v_s)$}
            \State $leadersStack.push(v_s)$
            \State $v \gets v_s$ 
            
        \EndIf
        
        \State $r \gets r -2$
        \EndWhile
        
        \State $\emph{committedRound} \gets v.round $
        \State $order\_vertices()$
        \Comment{see Algorithm~\ref{alg:commit}}
        
    \EndProcedure

\end{algorithmic}
\label{alg:ESBullshark}
\end{algorithm}

\section{Garbage collection in \sys}
\label{sec:GC}

One of the main practical challenges and a potential reason that DAG-based BFT protocols are not yet widely deployed is the need for unbounded memory to guarantee validity and fairness. In other words, the question of how to satisfy fairness and at the same time garbage collect old parts of the DAG from the working memory of the system. 

For example, HashGraph~\cite{baird2016swirlds} constructs an unstructured DAG, and thus has to keep in memory the entire prefix of the DAG in order to verify the validity of new blocks.
DAG-Rider\cite{keidar2021all}, Aleph~\cite{gkagol2019aleph}, and Narwhal~\cite{danezis2021narwhal} use a round-based structured DAG, but do not provide a solution to the aforementioned question.
The only DAG-based BFT we are aware of that proposed a garbage collection mechanism is Narwhal~\cite{danezis2021narwhal}.
Their mechanism uses the consensus decision in order to agree what rounds in the DAG can be cleaned.
However their protocol sacrifices the Validity (fairness) property of the BAB problem. 
It does not provide fairness to all parties since
blocks of slow parties can be garbage collected before they have a chance to be totally ordered. 
DAG-Rider\cite{keidar2021all}, on the other hand, make use of weak links to refer to yet unordered blocks in previous rounds, which guarantees that
every block is eventually ordered.
The solution works well in theory, but it is unclear how to garbage collect it.

In fact, through our investigation we realized that providing the BAB's validity (fairness) property with bounded memory in fully asynchronous executions is impossible since blocks of honest parties can be arbitrarily delayed. 
Similarly to the core observation in the FLP~\cite{fischer1985impossibility} impossibility result, in asynchronous settings, it is impossible to distinguish between faulty parties that will never broadcast a block and slow parties for which we need to wait before garbage collecting old rounds. 

\paragraph{Fairness after GST.} In the \sys implementation we propose a practical alternative.
We maintain bounded memory at the cost of providing fairness only after GST.
What we need is a $\Diamond P$ failure detector~\cite{chandra1996unreliable,larrea2004implementation} which will be strong and complete after GST letting us garbage collect
rounds even if we did not get vertices from all parities (i.e., we do not need to wait forever for faulty parties).
We do it by leveraging the structure of our DAG and introducing the notion of timestamp as described below. 
Formally, our implementation of \sys maintains bounded memory and satisfies the following:
\begin{definition}
\label{def:ESvalidity}
If an honest party $p_k$ calls $\textit{r\_bcast}_k(m,r)$ after GST, then every honest party $p_i$ eventually outputs $\textit{r\_deliver}_i(m,r,k)$.
\end{definition}

For the garbage collection mechanism we add a timestamp for every vertex. That is, an honest party specify in $v.ts$ the time when it broadcast its vertex $v$. 
In addition, parties maintain a garbage collection round, \emph{GCround}, and never add vertices to the DAG in rounds below it.
Note that the latency of the reliably broadcast building block we use is bounded after GST, but depends on the specific implementation. For the protocol description we assume that the time it takes to reliably broadcast a message after GST is $\Delta$.
The pseudocode, in which we describe how to change the function order\_vertices that is used by both versions of \sys, appears in Algorithm~\ref{alg:GC}.
The idea is simple. 
For every leader $v$ we order, we assign a timestamp $ts$, which is computed as the median of all the timestamp of $v$' parents (i.e., $v$'s strong edges).
Then, while traversing $v$'s causal history to find vertices to order, we compute a timestamp for every round in a similar way (the median of timestamps of the vertices in this round).
If the difference between the timestamp is above $3\Delta$ the round is garbage collected.

Since by the properties of the underling reliable broadcast all parties agree on the causal histories of the leaders, once parties agree which leaders to order  they also agree what rounds to garbage collect. 
Therefore, the garbage collection mechanism preserves the safety and liveness properties we prove in Appendix~\ref{app:proofs}. Below we argue that when announced with the above garbage collection, \sys satisfies Definition~\ref{def:ESvalidity} while preserving bounded memory.

\begin{algorithm}
\caption{Garbage collection. Algorithm for party $p_i$.}
\begin{algorithmic}[1]
 \footnotesize
 
   \Statex \textbf{Local variables:}
        \StateX $\emph{GCround} \gets 0$

    \Procedure{$order\_vertices()$}{}
        \While{$\neg \emph{leadersStack}.\emph{isEmpty}()$} 
        \State $v \gets \emph{leadersStack}.\emph{pop}()$ \label{alg:SMR:stackPop}
        \If {$v.round > 1$}
             \State $\emph{parents} \gets \{u \in DAG_i[v.round -1] ~|~ path(v,u)\}$
            \State $\emph{leaderTS} \gets
            \emph{median}(\{v.ts ~|~ v \in \emph{parents}\}$)
             \State $\emph{verticesToDeliver} \gets \emph{parents} \cup \{v\}$
        \Else 
            \State $\emph{verticesToDeliver} \gets \{v\}$
        \EndIf
        \State $r \gets \emph{GCround} +1$
        \While{$r < v.round -1$}
        
            \State $\emph{candidates} \gets \{u \in DAG_i[r] ~|~ path(v,u)\}$
            \State $\emph{candidatesTS} \gets 
            \emph{median}(\{v.ts ~|~ v \in \emph{candidates}\})$
            \State $\emph{verticesToDeliver} \gets \emph{verticesToDeliver} \cup  \emph{candidates} \setminus \emph{deliveredVertices}$
            \If{\emph{leaderTS} - \emph{candidatesTS} $> 3\Delta$}
                \State $\emph{GCround} \gets r$
                \State $DAG_i[r] \gets \{\}$
                \Comment{garbage collect old rounds}
            \EndIf
        
            \State $r \gets r+1$
        \EndWhile
        
        \For{$\textbf{every} ~v' \in \textit{verticesToDeliver}$}
        \Comment{in some deterministic order}
              \State \textbf{output}
              $\textit{a\_deliver}_i(v'.\textit{block},v'.\textit{round},
              v'.\textit{source})$
              \label{alg:SMR:decide}
              \State $\textit{deliveredVertices} \gets \textit{deliveredVertices} \cup \{v'\}$
          \EndFor
        \EndWhile
    \EndProcedure

    \Statex

    
    

\end{algorithmic}
\label{alg:GC}
\end{algorithm}

\paragraph{Bounded memory.}
In Appendix~\ref{app:proofs} we show that for every round $r$ there is a round $r' > r$ in which a leader is committed. In particular, this means that for every round $r$ with median timestamp $ts$, there will be eventually a committed leader with a high enough timestamp for $r$ to be garbage collected.

\paragraph{Fairness.}
First note that since every round has at least $2f+1$ vertices, the median timestamp of a round always belongs to an honest party.
Let $p_i$ be a party that broadcast a vertex $v$ at some round $r$ at time $t$ after GST, we show that all honest parties order $v$. By the assumption on the reliable broadcast latency, all honest parties reliably deliver $v$ before time $t + \Delta$.
Let $p_j$ be the first party that advances to round $r$.
In Appendix~\ref{app:proofs} we show that if an honest party advances to round $r$ at time $t$ after GST, then all honest parties advance to round $r$ no later than at time $t+ 2\Delta$.
Therefore, $p_j$ advanced to round $r$ not before $t - 2\Delta$.
Therefore, the timestamp of round $r$ is at least $t - 2\Delta$.
Thus, round $r$ is garbage collected only after a leader $v'$ with timestamp higher than $t + \Delta$ is ordered.
By the way the leader's timestamp is computed there is at least one vertex $v''$ in $v'.\emph{strongEdges}$ that broadcast by an honest party after time $t +\Delta$. 
Therefore, by the manner weak edges are added, there is an edge between $v''$ and $v$. Fairness follows since $v$ and $v''$ are in $v'$'s casual history and thus both ordered together with $v'$.  

\section{Implementation} \label{sec:implementation}
We implement a networked multi-core eventually synchronous \sys party forking the Narwhal project\footnote{\url{https://github.com/facebookresearch/narwhal}}. Narwhal provides the structured DAG used at the core of \sys, which we modify to support fast-path in partial synchrony as described in Section~\ref{sub:ourDAG}. Additionally, it provides well-documented benchmarking scripts to measure performance in various conditions, and it is close to a production system (it provides real networking, cryptography, and persistent storage). 
It is implemented in Rust, uses \texttt{tokio}\footnote{\url{https://tokio.rs}} for asynchronous networking, \texttt{ed25519-dalek}\footnote{\url{https://github.com/dalek-cryptography/ed25519-dalek}} for elliptic curve based  signatures, and data-structures are persisted using \texttt{Rocksdb}\footnote{\url{https://rocksdb.org}}. It uses TCP to achieve reliable point-to-point channels, necessary to correctly implement the distributed system abstractions.
By default, the Narwhal codebase runs the Tusk consensus protocol~\cite{danezis2021narwhal}; we modify the \texttt{proposer} module of the \texttt{primary} crate and the \texttt{consensus} crate to use \sys instead. Implementing \sys requires editing less than 200 LOC, and does not require any extra protocol message or cryptographic tool.
We are open-sourcing \sys\footnote{
\ifdefined\cameraReady
\url{https://github.com/asonnino/narwhal/tree/bullshark}
\else
https://www.dropbox.com/s/6cqulw8rosxl9bv/bullshark.zip?dl=0
\fi
} along with any Amazon web services orchestration scripts and measurements data to enable reproducible results\footnote{
\ifdefined\cameraReady
\url{https://github.com/asonnino/narwhal/tree/bullshark/benchmark/data}
\else
https://www.dropbox.com/s/6cqulw8rosxl9bv/bullshark.zip?dl=0
\fi
}.
\section{Evaluation} 
\label{sec:evaluation}
We evaluate the throughput and latency of our implementation of \sys through experiments on AWS. 
We particularly aim to demonstrate that 
(i) \sys achieves high throughput even for large committee sizes,
(ii) \sys has low latency even under high load, in the WAN, and with large committee sizes, and
(iii) \sys is robust when some parts of the system inevitably crash-fail. Note that evaluating BFT protocols in the presence of Byzantine faults is still an open research question~\cite{twins}.

We deploy a testbed on AWS, using \texttt{m5.8xlarge} instances across 5 different AWS regions: N. Virginia (us-east-1), N. California (us-west-1), Sydney (ap-southeast-2), Stockholm (eu-north-1), and Tokyo (ap-northeast-1). Parties are distributed across those regions as equally as possible. Each machine provides 10Gbps of bandwidth, 32 virtual CPUs (16 physical core) on a 2.5GHz, Intel Xeon Platinum 8175, 128GB memory, and runs Linux Ubuntu server 20.04. We select these machines because they provide decent performance and are in the price range of `commodity servers'.

In the following sections, each measurement in the graphs is the average of 2 independent runs, and the error bars represent one standard deviation; errors bars are sometimes too small to be visible on the graph. Our baseline experiment parameters are 10 honest parties, a maximum block size of 500KB, and a transaction size of 512B. We instantiate one benchmark client per party (collocated on the same machine) submitting transactions at a fixed rate for a duration of 5 minutes. The leader timeout value is set to 5 seconds. When referring to \emph{latency}, we mean the time elapsed from when the client submits the transaction to when the transaction is committed by one party. We measure it by tracking sample transactions throughout the system.

\subsection{Benchmark in the common case}
\Cref{fig:common-case} illustrates the latency and throughput of \sys, Tusk and HotStuff for varying numbers of parties.

\begin{figure*}[t]
\centering
\includegraphics[width=\textwidth]{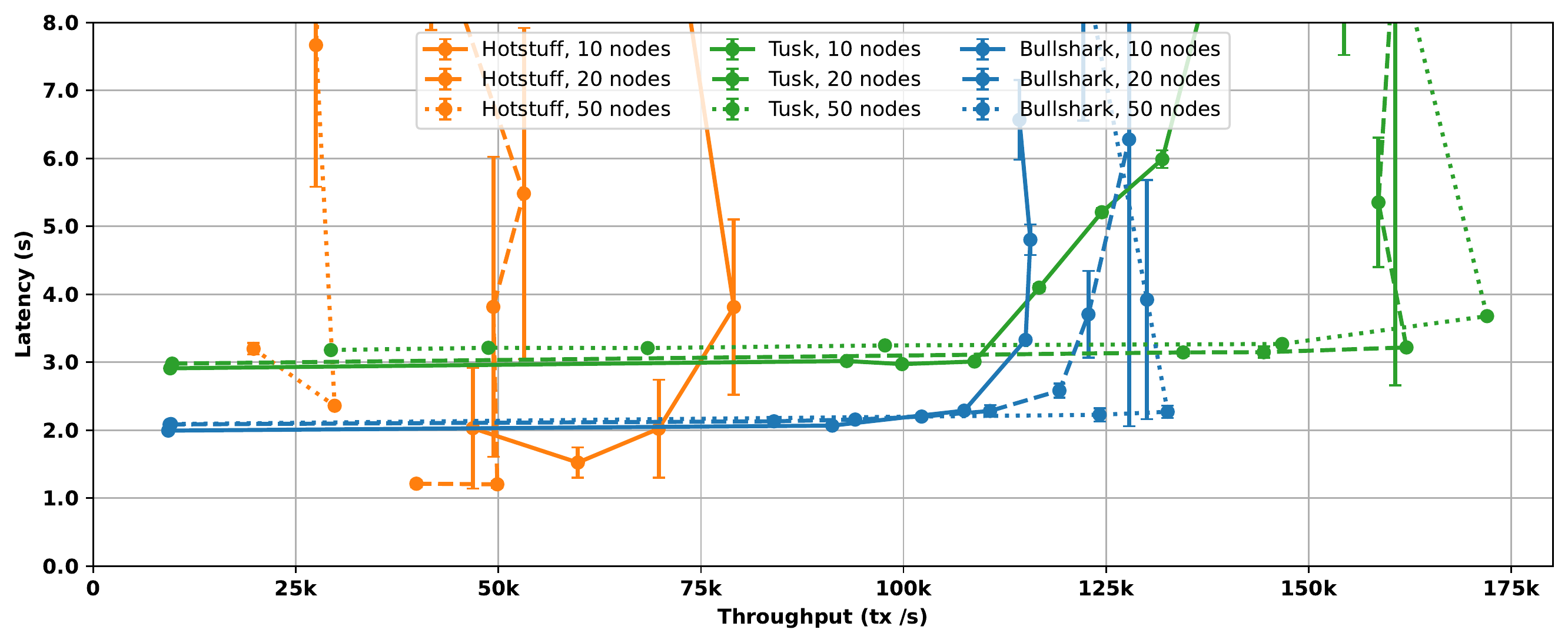}
\caption{
Comparative throughput-latency performance of HotStuff, Tusk, and \sys. WAN measurements with 10, 20, 50 parties. No faulty parties, 500KB maximum block size and 512B transaction size.
}
\label{fig:common-case}
\end{figure*}

\paragraph{HotStuff}
The maximum throughput we observe for HotStuff is 70,000 tx/s for a committee of 10 parties, and lower (up to 50,000 tx/s) for a larger committee of 20, and even lower (around 30,000 tx/s) for a committee of 50. The experiments demonstrate that HotStuff does not scale well when increasing the committee size. However, its latency before saturation is low, at around 2 seconds.

\paragraph{Tusk}
Tusk exhibits a significantly higher throughput than HotStuff. It peaks at 110,000 tx/s for a committee of 10 and at around 160,000 tx/s for larger committees of 20 and 50 parties. It may seem counter-intuitive that the throughput increases with the committee size: this is due to the implementation of the DAG not using all resources (network, disk, CPU) optimally. Therefore, more parties lead to increased multiplexing of resource use and higher performance~\cite{danezis2021narwhal}. Despite its high throughput, Tusk's latency is higher than HotStuff, at around 3 secs (for all committee sizes).

\paragraph{\sys}
\sys strikes a balance between the high throughput of Tusk and the low latency of HotStuff. Its throughput is significantly higher than HotStuff, reaching 110,000 tx/s (for a committee of 10) and 130,000 tx/s (for a committee of 50); \sys's throughput is over 2x higher than HotStuff's. 
Bullshark is built from the same DAG as Tusk and thus inherits its scalability allowing it to maintain high performance for large committee sizes. 
\sys's selling point over Tusk is its low latency, at around 2 sec no matter the committee size. \sys's latency is lower than Tusk since it commits within 2 DAG rounds while Tusk requires 4. \sys's latency is comparable to HotStuff and 33\% lower than Tusk. 
\Cref{fig:common-tps} highlights this trade-off by showing the maximum throughput that can be achieved by HotStuff, Tusk, and Bullshark while keeping the latency under 2.5s and 5s. Tusk and Bullshark scale better than HotStuff when increasing the committee size; there is no dotted line for Tusk since it cannot commit transactions in less than 2.5s. 

\begin{figure}[t]
\centering
\includegraphics[width=\columnwidth]{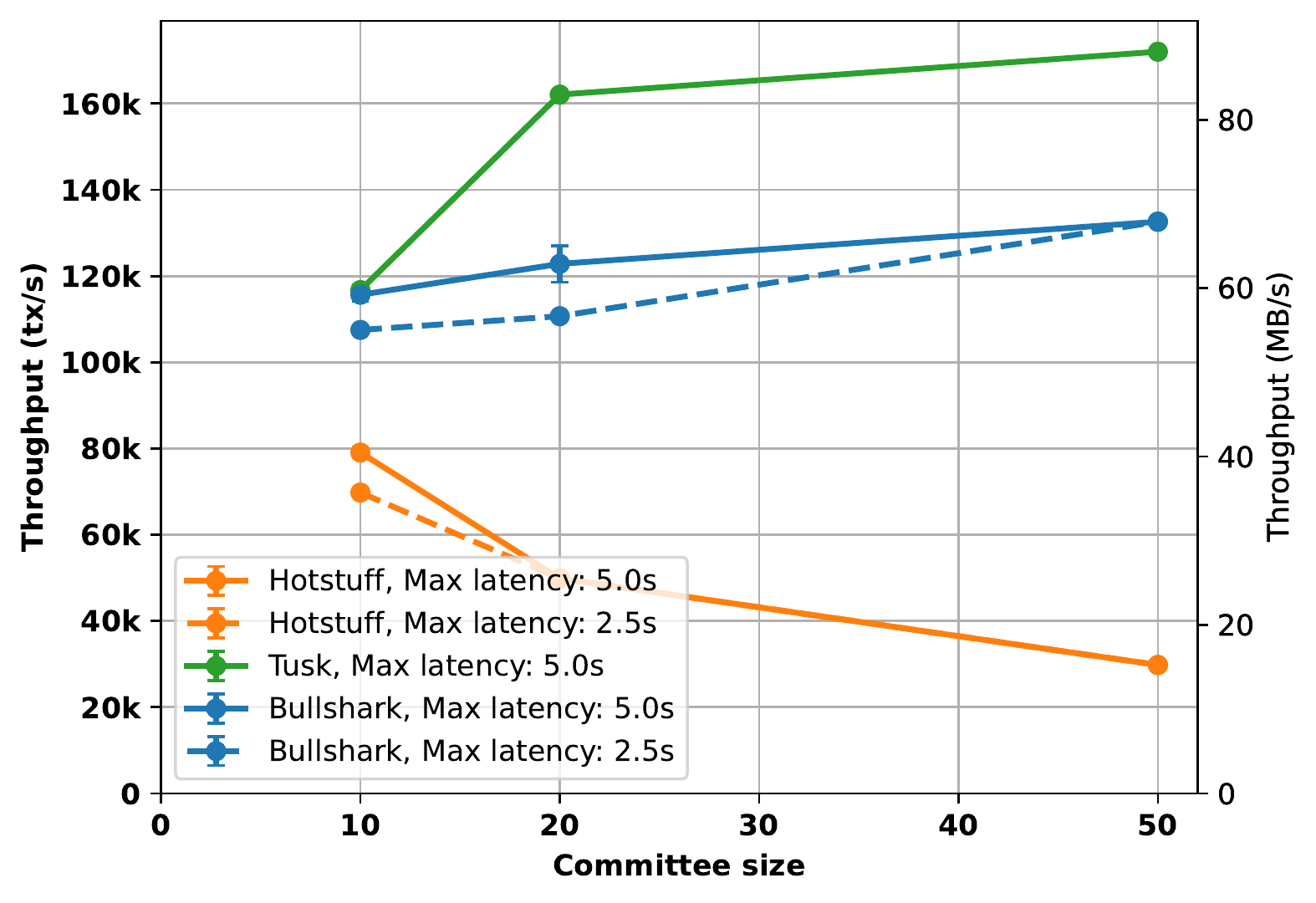}
\caption{
Maximum achievable throughput of HotStuff, Tusk, and \sys, keeping the latency under 2.5s and 5s. WAN measurements with 10, 20, 50 parties. No faulty parties, 500KB maximum block size and 512B transaction size.
 }
\label{fig:common-tps}
\end{figure}

\subsection{Benchmark under crash-faults}
\Cref{fig:faults} depicts the performance of HotStuff, Tusk, and \sys when a committee of 10 parties suffers 1 to 3 crash-faults (the maximum that can be tolerated in this setting). HotStuff suffers a massive degradation in throughput as well as a dramatic increase in latency. For 3 faults, the throughput of HotStuff drops by over 10x and its latency increases by 15x compared to no faults.
In contrast, both Tusk and \sys maintain a good level of throughput: the underlying DAG continues collecting and disseminating transactions despite the crash-faults, and is not overly affected by the faulty parties. The reduction in throughput is in great part due to losing the capacity of faulty parties. 
When operating with 3 faults, both Tusk and \sys provide a 10x throughput increase and about 7x latency reduction with respect to HotStuff.

\begin{figure}[t]
\centering
\includegraphics[width=\columnwidth]{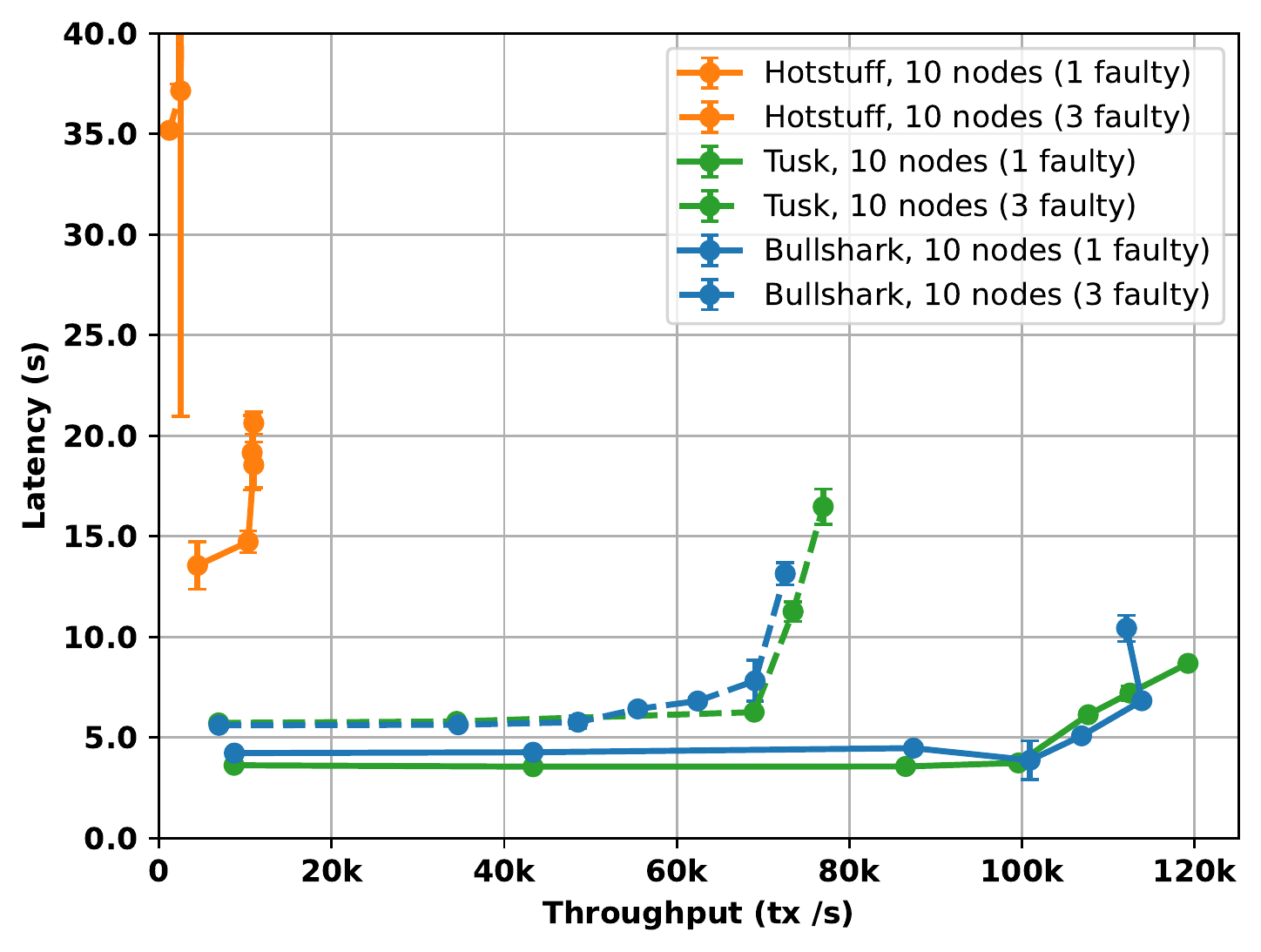}
\caption{
Comparative throughput-latency under crash-faults of HotStuff, Tusk, and \sys. WAN measurements with 10 parties. Zero, one, and three crash-faults, 500KB maximum block size and 512B transaction size.
}
\label{fig:faults}
\end{figure}

\subsection{Performance under asynchrony}
HotStuff has no liveness guarantees when the eventual synchrony assumption does not hold (before GST), either due to (aggressive) DDoS attacks targeted against the leaders~\cite{spiegelman2021ace} or adversarial delays on the leaders' messages as experimentally proven in prior work~\cite{danezis2021narwhal,gelashvili2021jolteon} . That is, the throughput of the system falls to $0$. The same can happen to the partially synchronous version of \sys. The reason is that whenever a party becomes the leader for some round, its proposal can be delayed such that all other parties timeout for that round. In order to avoid this attack, Tusk and DAG-Rider elects leaders unpredictably after the DAG is constructed which makes such attacks impossible.
The purpose of the fallback mode of \sys is to maintain the same liveness properties as Tusk and DAG-Rider under asynchrony without compromising on performance during periods of synchrony. If the voting type of all parties is fallback, then \sys acts as Tusk. In the fallback mode, \sys thus renounces to its latency advantage with respect to Tusk in order to remain live under asynchrony. As any asynchronous protocol, the performance of both Tusk and \sys during periods of asynchrony can be arbitrarily bad as they depend on the network conditions (which guarantee delivery after unbounded time). When the period of asynchrony ends, parties change their voting type to steady-state, and \sys offers again its state-of-the-art latency.

\section{Related work} \label{sec:related}

In this Section we discuss other prior works relevant to \sys and a more in depth comparison with the systems against which we evaluate.

\paragraph{Performance comparisons:}
We compare \sys with Tusk~\cite{danezis2021narwhal} and HotStuff~\cite{hotstuff}. Tusk is the most similar system to \sys. It is a zero-message consensus protocol built on top of the same structured DAG as \sys. It is however fully asynchronous while \sys is partially-synchronous fast path. HotStuff is an established partially-synchronous protocol running at the heart of a number of projects~\cite{diem, celo, flow, thunder, cypherium}, and a successor of the popular Tendermint~\cite{tendermint}.

We aim to compare \sys with related systems as fairly as possible. An important reason for selecting Tusk\footnote{https://github.com/asonnino/narwhal} and HotStuff\footnote{https://github.com/asonnino/hotstuff} is because they both have open-source implementations sharing deep similarities with our own. They are both written in Rust using the same network, cryptographic and storage libraries than ours. They are both designed to take full advantage of multi-core machines and to run in the WAN. 

We limit our comparison to these two systems, thus omitting a number of important related works such as~\cite{guo2020dumbo, stathakopoulou2019mir, chan2020streamlet, castro1999practical,kogias2016enhancing, tendermint, yang2019prism}. 
A practical comparison with those systems would hardly be fair as they do not provide an open-source implementations comparable to our own. Some selected different cryptographic libraries, use different cryptographic primitives (such as threshold signatures), or entirely emulate all cryptographic operations. A number of them are written in different programming languages, do not provide persistent storage, use a different network stack, or are not multi-threaded thus under-utilizing the AWS machines we selected. Most implementations of prior works are not designed to run in the WAN (e.g., have no synchronizer), or are internally sized to process empty transactions and are thus not adapted to the 512B transaction size we use. Instead, we provide below a discussion on the performance of alternatives based on their reported work.

\paragraph{Partially-synchronous protocols:}
Hotstuff-over-Narwhal~\cite{danezis2021narwhal} and Mir-BFT \cite{mir-bft} are the most performant partially synchronous consensus protocols available.
The performance of the former is close to \sys under no faults given that they share the same mempool implementation. However, \sys performs considerably better under faults and the engineering effort of Hotstuff-over-Narwhal is double that of \sys. The extra code required to implement \sys over Narwhal is about 200 LOC\footnote{
\ifdefined\cameraReady
\url{https://github.com/asonnino/narwhal/tree/bullshark}
\else
https://www.dropbox.com/s/6cqulw8rosxl9bv/bullshark.zip?dl=0
\fi
} 
(Alg.~\ref{alg:ESBullshark}) whereas the extra code of Hotstuff is more than 4k LOC. Additionally, \sys adapts to an asynchronous environment with the fallback protocol unlike Hotstuff that will completely forfeit liveness during asynchrony leading to an explosion of the confirmation latency (see \Cref{fig:faults} of \Cref{sec:evaluation}).

For Mir-BFT with transaction sizes of about 500B (similar to our benchmarks), the peak performance achieved on a WAN for 20 parties is around 80,000 tx/sec under 2 seconds -- a performance comparable to our baseline HotStuff. Impressively, this throughput decreases only slowly for large committees up to 100 nodes (at 60,000 tx/sec). Crash-faults lead to throughput dropping to zero for up to 50 seconds, and then operation resuming after a reconfiguration to exclude faulty nodes. \sys offers higher performance (almost 2x), at the same latency. 


\paragraph{DAG-based protocols:}
The DAG have been used in the context of Blockchains  in multiple systems. 
Hashgraph~\cite{baird2016swirlds} embeds an asynchronous consensus
mechanism into a DAG without a round-by-round step structure which results to unclear rules on when consensus is reached. This consequently results on an inability to implement garbage collection and potentially unbounded state.
Finally, Hashgraph uses local coins for randomness, which can potentially
lead to exponential latency. 

A number of blockchain projects build consensus over a DAG under open participation, partial synchrony or asynchrony network assumptions. GHOST~\cite{sompolinsky2015secure} proposes a finalization layer over a proof-of-work consensus protocol, using sub-graph structures to confirm blocks as final potentially before a judgment based on longest-chain / most-work chain fork choice rule can be made.
Tusk~\cite{danezis2021narwhal} is the most similar system to \sys. It is an asynchronous consensus using the same structured DAG as \sys. 
A limitation of any reactive asynchronous protocol, such as Tusk, is that slow parties are indistinguishable from faulty ones, and as a result the protocol proceeds without them. This creates issues around fairness and incentives, since honest, but geographically distant authorities may never be able to  commit transactions submitted to them.
%
%
Further, Tusk relies on clients to re-submit a transaction if it is not sequenced in time, due to leaders being faulty.
%
In contrast, both versions of \sys satisfy fairness after GST while ensuring bounded memory via a garbage collection mechanism.   

\paragraph{Dual-Mode Consensus Protocols:}
The idea of having optimistic and fallback paths in BFT consensus has first been explored by Kurasawe et al~\cite{kursawe2005optimistic} with followup improvements~\cite{ramasamy2005parsimonious,spiegelman2020search} on the communication complexity. However, these papers are theoretical and not designed for high-load applications hence their implementation would at best be close to the Hotstuff baseline.

The seminal work from Guerraoui et al~\cite{guerraoui2010next} introduced Abstract, a framework in which developers can plug and play multiple consensus protocols based on the environment they plan to deploy the protocol. A followup work called the Bolt-Dumbo Transformer (BDT)~\cite{lu2021bolt}, can be seen as instantiating of Abstract for the specific use case of a dual-mode consensus protocol.
BDT takes Abstract's general proposal and instantiates it by composing three separate consensus protocols as black boxes. Every round starts with 1) a partially synchronous protocol (HotStuff), times-out the leader and runs 2) an Asynchronous Binary Agreement in order to move on and run 3) a fully asynchronous consensus protocol~\cite{guo2020dumbo} as a fallback.
Ditto~\cite{gelashvili2021jolteon} follows another approach that does not require these black boxes.
Instead, it combines a 2-phase variant of Hotstuff with a variant of the asynchronous VABA~\cite{abraham2019asymptotically} protocol for fallback.
As a result it reduces the latency cost of BDT significantly, but cannot be generalized to a plug-and-play framework.


All the protocols above solve the problem of consensus in asynchrony, but they include the actual transactions in the proposals, hence their throughput is bounded by the one of Hotstuff.
A way to increase their throughput would be to adopt the Narwhal-HS~\cite{danezis2021narwhal} approach introduced in prior work, which substitute the transaction dissemination with Narwhal as a mempool and includes only hashes of mempool batches in the proposals. This would potentially achieve similar performance to \sys. However it would come at the steep costs of maintaining two code-bases (one for the mempool and one for the consensus), higher latency (since Narwhal does a reliable broadcast which is usually the first step of a consensus protocol) and loss of quantum-safety (since they all use threshold signatures to provide Safety with lower communication complexity). 
Unlike these ``hybrids'', \sys provides both the theoretical contribution of being the first BAB with all the good properties we already described, the practical contribution of significant latency gains in synchrony and the usability contribution of modifying only 200 LOC from the base-protocol Tusk.

\section{Discussion}
\label{sec:conclusion}
On the foundational level \sys is the first DAG-based zero overhead BFT protocol that achieves the best of both worlds of partially synchronous and asynchronous protocols. 
It keeps all the desired properties of DAG-Rider, including optimal amortized complexity, asynchronous liveness, and post quantum security, while also allowing a fast-path during periods of synchrony.
\sys's parties switch their voting type to fallback after every unsuccessful wave.
An interesting future direction is to add an adaptive mechanism for parties to learn when is best to switch between the types.
Interestingly, since the DAG provides full information, this mechanism can be also implemented without extra communication.

The partially synchronous version of \sys is extremely simple (200 LOC) and highly efficient.
In particular, it does not need any view-change or view-synchronization mechanisms since the DAG already encodes all the required information.  
When implemented over the Narwhal mempool it has $2x$ the throughput of the partially synchronous HotStuff protocol and $33\%$ lower latency than the asynchronous Tusk protocol over Narwhal. 



\ifdefined\cameraReady
  \section*{Acknowledgements}
  This work was initiated when the authors were part of Novi reseacrh at Facebook. We thank George Danezis for his valuable feedback.
\fi


\bibliographystyle{ACM-Reference-Format}
\bibliography{bib}

\appendix

\section{Partially Synchronous Bullshark Illustration}
\label{app:PSB}

Figure~\ref{fig:PSBullshark} illustrates the partially synchronous Bullshark protocol for $n = 4$ and $f = 1$.
Each odd round in the DAG has a predefined leader vertex (highlighted in solid green) and the goal is to first decide which leaders to commit. Then, to totally order all the vertices in the DAG, a party goes one by one over all the committed leaders and deterministically orders their causal histories.

Each vertex in an even round can contribute one vote for the previous round leader. In particular, a vertex in round r votes for the leader of round $r-1$ if there is an edge between them. The commit rule is simple: a leader is committed if it has at least $f+1$ votes. In Figure~\ref{fig:PSBullshark}, L3 is committed with 3 votes, whereas L1 and L2 have less then $2=f+1$ votes and are not committed.

Due to the asynchronous nature of the network, the local views of the DAG might differ for different parties. That is, some vertices might be delivered and added to the local view of the DAG of some of the parties but not yet delivered by the others. Therefore, even though some validators have not committed L1, others might have.

To guarantee all parties commit the same leaders, Bullshark relies on quorum intersection: 

\begin{quote}
\textit{Since the commit rule requires $f+1$ votes and each vertex in the DAG has at least $n-f$ edges to vertices from the previous round, it is guaranteed that if some validator commits a leader L then all future leaders will have a path to at least one vertex that voted for L, and thus will have a path to L.}
\end{quote}

Therefore:
 \textit{\textbf{If there is no path to a leader L from a future leader, then no party committed L and it is safe to skip L.}}
 
 The logic to order leaders is the following: when a leader $i$ is committed, the party checks if there is a path between leader $i$ to leader $i-1$. 
 If this is the case, leader $i-1$ is ordered before leader $i$ and the logic is recursively restarted from $i-1$.
 
 Otherwise, leader $i-1$ is skipped and the party checks if there is a 
 path between $i$ to $i-2$.
 If there is a path, leader $i-2$ is ordered before $i$ and the logic is recursively restarted from $i-2$. 
 Otherwise, leader $i-2$ is skipped and the process continues in the same way.
 The process stops when it reaches a leader that was previously ordered.
 
 In Figure~\ref{fig:PSBullshark}, leaders L1 and L2 do not have enough votes to be committed and once the party commits L3 it has to decide whether to order L1 and L2. Since there is no path from L3 to L2, L2 can be skipped. However, since there is a path between L3 and L1, L1 is ordered before L3.
 Now, to totally order the vertices of the DAG, the party first orders the causal history of L1 (nothing to order in this example) by some deterministic rule and then orders the causal history of L3.

\begin{figure}
    \centering
    \includegraphics[width=0.48\textwidth]{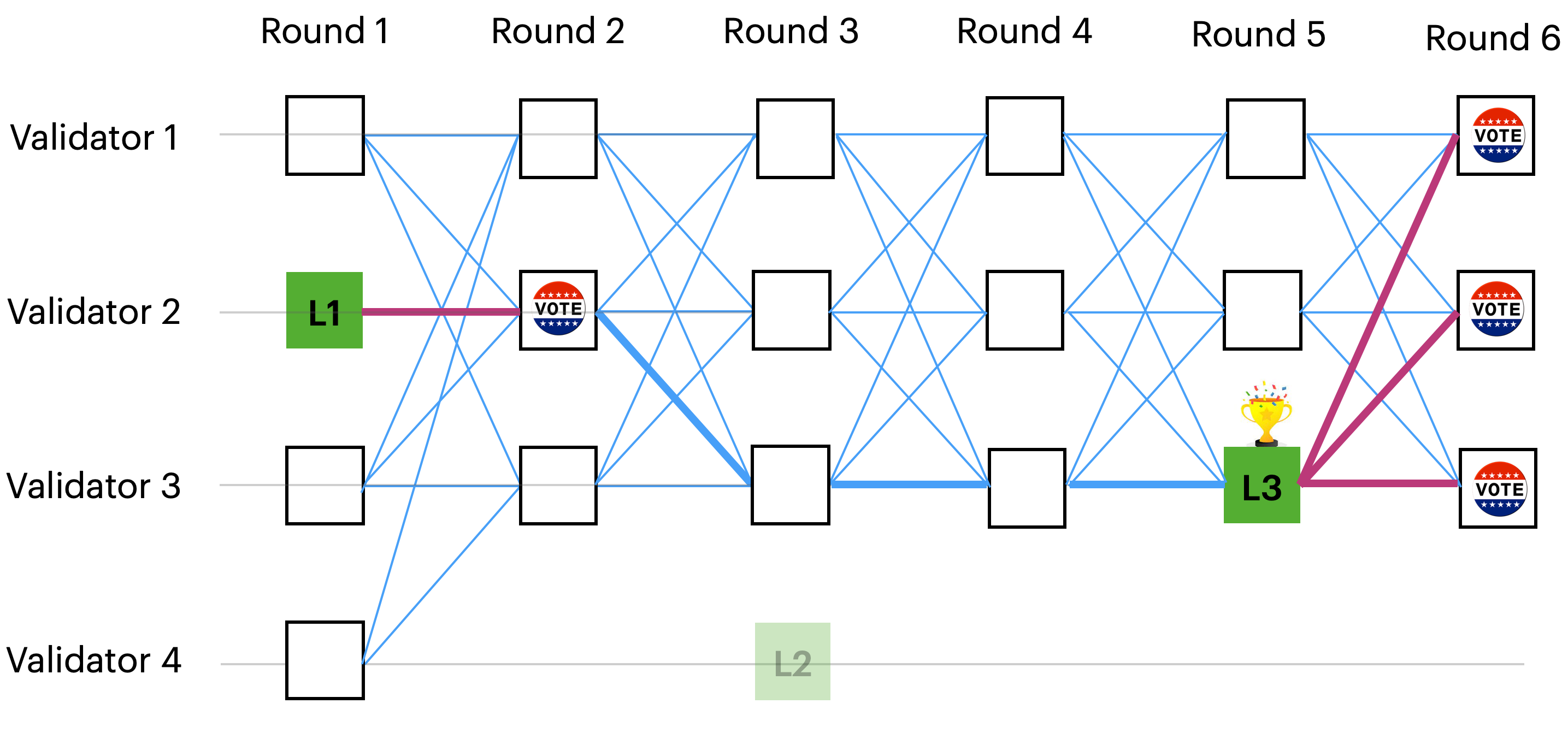}
    \caption{Illustration of the partially synchronous Bullshark.}
    
    \label{fig:PSBullshark}
\end{figure}
\section{Logical vs physical DAG}
\label{app:logicaldag}

As mentioned above, to provide deterministic fast path, introducing timeouts is unavoidable~\cite{fischer1985impossibility}.
After implementing and evaluating two alternatives, we decided to embed the timeouts into the DAG construction as described above.
Intuitively, it might look inefficient as the DAG does not advance in network speed, but as we shorty explain, it is the other way round. 

The other approach we consider is a virtual consensus DAG layer on top of the physical DAG.
In this case the physical level has no timeouts and is very similar to the DAG construction in DAG-Rider, which advances rounds in networks speed once
$2f+1$ nodes in the current round are delivered.
To encode timeouts, some of the nodes in the physical DAG have ``consensus’’
headers indicating that they belong to the virtual level.
The logic to advance consensus rounds is almost similar to the one
described in Alg~\ref{alg:DAG}.
That is, consensus nodes indicate in their consensus header to which virtual
nodes they refer as parents. This virtual nodes can be in arbitrary physical DAG rounds but they are at exactly one less ($r-1$) consensus round. As a result, now timeouts are only needed at the virtual level and do not interfere with the physical DAG advancement.
The only difference from Alg~\ref{alg:DAG} is that weak links are not required on the virtual
level since the weak links on the physical level already guarantee the
validity property.
All in all, the physical DAG advances in network speed and the virtual DAG  provides the functionality required by the \sys consensus protocol. 

We implemented and evaluated this logical DAG construction, however, the results were not encouraging (around 50\% latency increase without any significant throughput benefit). After investigation we attributed this to two main reasons:
\begin{itemize}
    \item Since \sys is built on top of Narwhal, it inherent the data dissemination decoupling from the DAG construction. That is, data is disseminated at network speed regardless of the DAG construction, which contains only metadata. Therefore, if the DAG advances rounds slower, then each vertex in the DAG simply contains more metadata and the throughput is not compromised.

    \item The logical split between virtual and physical DAG introduces a decoupling between delays/timeouts for the consensus messages and delays for the block creation. This results to a common pattern where a physical DAG blocks is created milliseconds before a vote is ready to be cast, but the vote missed the block and needs to wait for the next round to be cast. This introduces a small delay per vote but since we need 2f+1 votes to commit a consensus round the latency of the DAG moves from the median latency to the tail-latency of the 66th percentile. 
    
    \item Moreover, the smaller the DAG the less resources are required to manage it. For example, less memory to store it and less bandwidth to construct it.
\end{itemize}

\section{Proofs}
\label{app:proofs}

We provide proofs of correctness for both versions of \sys. 

\subsection{\sys With Fallback}
\label{sec:fallback}

\paragraph{Total order.}
Note that at any given time parties might have slightly different local DAGs.
This is because some vertices may be delivered at some parties but not yet at others.
However, since we use reliable broadcast for each vertex $v$, and wait for the entire causal history of $v$ to be added to the before we add $v$, we get the following important observation:

\begin{observation}
\label{obs:relaible}
For every two honest parties $p_i$ and $p_j$ we get:
\begin{itemize}
    \item For every round $r$, $\bigcup_{r > 0} DAG_i[r]$ is eventually equal to\\ $\bigcup_{r > 0} DAG_j[r]$.
    \item For any given time $t$ and round $r$, if $v \in DAG_i[r]$ $\wedge$  $v' \in DAG_i[r]$ s.t.\ $v.source = v'.source$, then $v = v'$.
    Moreover, for every round $r'<r$, if $v'' \in DAG_i[r']$ and there is a path from $v$ to $v''$, then $v'' \in DAG_j[r']$ and there is a path between $v'$ to $v''$. 
\end{itemize}
\end{observation}

To totally order the vertices in the DAG, each party $p_i$ locally interprets $DAG_i$ (there is no extra communication on top of building the DAG).
To this end, $p_i$ divides its DAG into waves of 4 rounds each. 
Every wave has 3 leaders that can potentially be committed: 2 steady-state leaders and one fallback leader.
The steady-state leaders are two pre-defined vertices, one in the first round of the wave and the other in the third.
The fallback leader is a vertex in the first round of the wave that is selected by the randomness produced in the fourth round of the wave. 
To make sure a fallback leader and a steady state leader are not committed in the same wave, each party can only vote for either the fallback leader or the steady-state ones.
In the code, \emph{steadyVoters[w]} \emph{fallbackVoters[w]} contain all the parties that can vote for steady-state or fallback leaders in wave $w$, respectively.
We say that a party $p_i$ \emph{determines} $p_j$ \emph{vote type} to be a steady-state (fallback)  in wave $w$ if its $\emph{steadyVoters}[w]$ ($fallbackVoters[w]$) contains $p_j$. 
Moreover, as we show in the next claim, all parties agree on $p_j$'s vote type in wave $w$.
This, in particular, means that Byzantine parties cannot equivocate or hide their vote (a nice property that we get from using reliable broadcast as a building block). 
\begin{claim}
\label{claim:votetype}

For every party $p_i$ and round $r$, each party $p_j$ determines at most one vote type for $p_i$ in wave $w$.
Moreover if $p_j$ and $p_k$ determine vote type $T$ and $T'$ for $p_i$ in wave $w$, respectively, then $T = T'$.  

\end{claim}
\begin{proof}
The first part of the claim follows from the code of function $\emph{determine\_party\_vote\_type}$.
This function is called by a party $p_j$ whenever it adds a new vertex $v$ to $DAG_j[r]$ such that $r$ is the first round of a wave, and the source of the vertex (a party $p_i$) is either added to \emph{steadyVoters[w]} or \emph{fallbackVoters[w]}.
The second part of the claim follows from Observation~\ref{obs:relaible} and the fact (by the code of try\_add\_to\_DAG) that $v$ is added to the DAG only after all its causal history is added.
This guarantees that for every wave $w$ and party $p_i$ try\_steady\_commit and try\_fallback\_commit are called with the same parameters and thus return the same result.
This in turn guarantees that all parities that determine $p_i$'s vote type in wave $w$ see the same type. 
\end{proof}

\vspace{2mm}
There are two possible ways to commit a leader $v$ in \sys.
The first is to \emph{directly} commit it when either try\_steady\_commit or try\_fallback\_commit, called with $v$, return true.
The second option is to \emph{indirectly} commit it when it is added to \emph{leaderStack} in Line~\ref{line:stacksteady} or~\ref{line:stackfallback}.
In both cases, to commit a leader in wave $w$, we count the number of vertices in some round (depending on the leader type and whether we directly or indirectly commit it) in $w$ that have a strong path to the leader and their vote corresponds to the leader's type. 
We first show that steady state and fallback leaders cannot be directly committed in the same wave.

\begin{claim}
\label{claim:samewave}

If a party $p_i$ directly commits a steady-state leader in wave $w$, then no party commits (directly or indirectly) a fallback leader in wave $w$, and vice versa.

\end{claim}

\begin{proof}
Consider a steady state leader vertex $v$ committed by a party $p_i$ in round $r$ in wave $w$.
By the code, to directly commit a leader vertex a party need to determine the vote type of at least $2f+1$ parties in the wave to be the same as the leaders.
Similarly, to indirectly commit a vertex leader, a party needs to determine the vote type of at least $f+1$ parties in the wave to be the same as the leaders.
Since $p_i$ directly commits state leader vertex $v$ in wave $w$, it determines $2f+1$ parties as steady state voters in wave $w$.
Since there are $3f+1$ parties in total, by Claim~\ref{claim:votetype}, no other party determines more than $f$ parties as fallback voters in wave $w$.
Therefore, no other party commit (directly or indirectly) a fallback leader in wave $w$.
From symmetry, the same argument works in the other direction.
\end{proof}

\vspace{2mm}
For the proof of the next lemmas we say that a party $p_i$ \emph{consecutively directly commit} leader vertices $v_i$ and $v'_i$ if $p_i$ directly commits them in rounds $r_i$ and $r'_i > r_i$, respectively, and does not directly commit any leader vertex between $r_i$ and $r'_i$. 
In the next claims we are going to show that honest parties commit the same leaders and in the same order:

\begin{claim}
\label{claim:commitmin}

Let $v_i$ and $v'_i$ be two leader vertices consecutively directly committed by a party $p_i$ in rounds $r_i$ and $r'_i > r_i$, respectively.
Let $v_j$ and $v'_j$ be two leader vertices consecutively directly committed by a party $p_j$ in rounds $r_j$ and $r'_j > r_j$, respectively.
If $r_i \leq r_j \leq r'_i$, then both $p_i$ and $p_j$ (directly or indirectly) commit the same leader in round $min(r'_i,r'_j)$.

\end{claim}

\begin{proof}
Claim~\ref{claim:samewave} implies that that there is at most one committed leader in each round.
Thus, if $r'_i = r'_j$ we are done.
Otherwise, assume without lost of generality that $r'_i < r'_j$.
Thus, if $r_j = r'_i$ we are done.
Otherwise, we need to show that $p_j$ indirectly commits $v'_i$ in $r'_i$.

By the code of commit\_leader, after $p_j$ directly commits $v_j'$ in round $r'_j$ it tries to indirectly commit leaders in round numbers smaller than $r'$ until it reaches round $r_j < r'_i$.
Let $r'_i < r < r'_j$, be the smallest number between $r'_i$ and $r'_j$ in which $p_j$ (directly or indirectly) commits a leader $v$.
Consider two cases:
\begin{itemize}
    \item Vertex $v'_i$ is a steady-state leader.
    Note that $r > r'_i + 1$ since only odd rounds have potential leaders.
    Since $p_i$ directly commits $v'_i$ in round $r'_i$, there is a set $C$ of $2f+1$ vertices in $DAG_i[r'_i + 1]$ with strong paths to $v'_i$ and with $v'_i$'s types. By observation~\ref{obs:relaible}, Claim~\ref{claim:votetype}, and quorum intersection, there are at least $f+1$ vertices in $DAG_j[r'_i + 1]$ with $v'_i$'s vote type and strong paths from the $v$ to them.
    
    \item Vertex $v'_i$ is a fallback leader. 
    By Claim~\ref{claim:samewave}, no leader is committed in round $r+2$.
    Thus, $r > r'_i + 3$.
    Since $p_i$ directly commits $v'_i$ in round $r'_i$ and $v_i$, there is a set $C$ of $2f+1$ vertices in $DAG_i[r'_1 + 3]$ with strong paths to $v'_i$ and with $v'_i$'s types. 
    By observation~\ref{obs:relaible}, Claim~\ref{claim:votetype}, and quorum intersection, there are at least $f+1$ vertices in $DAG_j[r'_i + 3]$ with $v'_i$'s vote type and strong paths from the $v$ to them.
\end{itemize}
In both cases $p_j$ counts (in \emph{ssVotes} or \emph{fbVotes}) at least $f+1$ votes for the leader.
In addition, by observation~\ref{obs:relaible} and Claim~\ref{claim:votetype}, since in both cases there are at least $2f+1$ vertices with the $v$'s type, there are at most $f$ vertices with the opposite type. 
Thus, $p_j$ counts at most $f$ votes for the other leader.
Therefore, by Lines \ref{line:indirectfirst}-\ref{line:indirectlast} in commit\_leader, $p_j$ indirectly commits $v'_i$.


\end{proof}

\begin{claim}
\label{claim:commitinterval}

Let $v_i$ and $v'_i$ be two leader vertices consecutively directly committed by a party $p_i$ in rounds $r_i$ and $r'_i > r_i$, respectively.
Let $v_j$ and $v'_j$ be two leader vertices consecutively directly committed by a party $p_j$ in rounds $r_j$ and $r'_j > r_j$, respectively.
Then $p_i$ and $p_j$ commits the same leaders between rounds $max(r_i,r_j)$ and $min(r'_i,r'_j)$, and in the same order.

\end{claim}

\begin{proof}
If $r'_i < r_j$ or $r'_j < r_i$, then we are trivially done because there are no rounds between $max(r_i,r_j)$ and $min(r'_i,r'_j)$.
Otherwise, assume without lost of generality that $r_i \leq r_j \leq r'_i$.
By Claim~\ref{claim:commitmin}, both $p_i$ and $p_j$ (directly or indirectly) commit the same leader in round $min(r'_i,r'_j)$.
Assume without lost of generality that $min(r'_i,r'_j) = r'_i$.
Thus, by Claim~\ref{claim:samewave}, both $p_i$ and $p_j$ commit $v'_i$ in round $r'_i$ and $v_j$ in round $r_j$.
By the code of commit\_leader, after (directly or indirectly) committing a leader, parties try to indirectly commit leaders in smaller round numbers until they reach a round in which they previously directly committed a leader.
Therefore both $p_i$ and $p_j$ will try to indirectly commit all leaders going down from $r'_i = min(r'_i,r'_j)$ to $r_j = max(r_i,r_j)$.
Since $v'_i$ appears in both $DAG_i$ and $DAG_j$, by Observation~\ref{obs:relaible}, all vertices in $DAG_i$ such that there is a path from $v'_i$ to them appear also in $DAG_j$.
The claim follows from the deterministic code of the function commit\_leader.

\end{proof}

\vspace{2mm}
By inductively applying Claim~\ref{claim:commitinterval} for every pair of honest parties  we get the following:

\begin{corollary}
\label{col:order}

Honest parties commit the same leaders and in the same order.

\end{corollary}

For the next lemma we say that the \emph{causal history} of a vertex leader $v$ in the DAG is the set of all vertices such that there is a path from $v$ to them.

\begin{lemma}
\label{lem:order}

Algorithms~\ref{alg:dataStructures}, \ref{alg:DAG}, \ref{alg:modes}, and~\ref{alg:commit} satisfy Total order.

\end{lemma}

\begin{proof}
By Corollary~\ref{col:order}, honest parties commit the same leaders and in the same order.
By the code of the order\_vertices procedure, parties iterate on the committed leaders according to their order and  \emph{a\_deliver} all vertices in their causal history by a pre-defined deterministic rule.
The lemma follows by Observation~\ref{obs:relaible} since all honest parties has the same casual history in their DAG for every committed leader.

\end{proof}

\paragraph{Agreement and Validity.}

\begin{lemma}
\label{lem:agreement}

Algorithms~\ref{alg:dataStructures}, \ref{alg:DAG}, \ref{alg:modes}, and~\ref{alg:commit} satisfy Agreement.

\end{lemma}

\begin{proof}
Assume some honest party $p_i$ outputs $a\_deliver(v_i.block,\\ v_i.round, v_i.source)$.
We will show that every honest party $p_j$ outputs it as well.
By the code of order\_vertices, there is a leader vertex $v$ that $p_i$ committed such that $v_i$ is in $v$'s casual history.
By Observation~\ref{obs:relaible}, the $v$'s casual histories in $DAG_i$ and $DAG_j$ are the same.  
Thus, by code of order\_vertices, we only need to show that $p_j$ eventually commit leader vertex $v$.
Let $v'$ be the leader vertex with the lowest number that is higher than $v.round$ that $p_i$ directly commits.
Let $v''$ be the vertex that triggers this direct commit, i.e., the vertex $v''$ that passed to the try\_ordering function that calls determine\_party\_vote\_type, which in turn commits $v'$.
By Observation~\ref{obs:relaible}, $p_j$ eventually add $v''$ to $DAG_j$ and call try\_ordering with $v''$.
By Observation~\ref{obs:relaible} again, the casual history of $v''$ in $DAG_i$ is equivalent the casual history of $v''$ in $DAG_j$.
Hence, $p_j$ directly commits $v'$ as well.
Since the casual history of $v'$ in $DAG_i$ is also equivalent the casual history of $v'$ in $DAG_j$, $p_j$ also commits $v$.

\end{proof}

\vspace{2mm}
By the Liveness (Agreement and Validity) properties of reliable broadcast and since it is enough for parties to deliver $2f+1$ vertices in a round in order to move to the next one, the DAG grows indefinitely:

\begin{observation}
\label{obs:DAGgrows}

For every round $r$ and honest party $p_i$, $DAG_i[r]$ eventually contains a vertex for every honest party.

\end{observation}

In the next to claims we show that for every round $r$ there is an honest party $p_i$ that commit a leader in a round higher than $r$ with probability $1$. First, we show that if it is not the case, then starting from some point the vote type of all parties is fallback. Note that this is true also for Byzantine parties since thanks to the reliable broadcast Byzantine parties cannot lie about their casual history.

\begin{claim}
\label{claim:honestfallback}

Consider an honest party $p_i$.
If there is a wave $w$ after which no honest party commits a leader, then in all waves $w' >w +1$ $p_i$ determines the vote type of all parties that reach $w'$ in $DAG_i$ as fallback. 

\end{claim}

\begin{proof}
Let $w' > w + 1$ be a wave that start after $w$. By the claim assumption no honest party commits a leader in wave $w' -1$. Let $r$ be the first round of wave $w'$. 
Consider a party $p_j$ for which $p_i$ has a vertex $v_j$ in $DAG_i[r]$
By the code, $p_i$ calls try\_ordering with $v_j$, which in turn calls determine\_party\_vote\_type to determine $p_j$'s vote type for $w'$.
By Observation~\ref{obs:relaible}, the casual history of $v_j$ in $DAG_j$ is equivalent to the casual history of $v_j$ in $DAG_i$.
The claim follows from the code of determine\_party\_vote\_type.
Since $p_j$ did not commit a leader in wave $w'$, both functions try\_steady\_commit and try\_fallback\_commit return falls $p_i$' invocation of\\ determine\_party\_vote\_type.
Therefore, $p_i$ sets $p_j$'s vote type in $w'$ to fallback.

\end{proof}

\vspace{2mm}
The following claim is a known property of all to all communication, which sometimes referred as \emph{common core}~\cite{canetti1996studies}.
We provide proof for completeness.

\begin{claim}
\label{claim:core}

For every wave $w$ and party $p_i$.
Let $r$ be the first round of $w$.
If $|DAG_i[r+k]| \geq 2f+1, k \in \{0,1,2,3\}$, then there is a set $C \subseteq DAG_i[r]$ such that $|C| = 2f+1$ and for every vertex $v \in C$ there are $2f+1$ vertices in $DAG_i[r+3]$ with strong paths to $v$.

\end{claim}

\begin{proof}
The proof follows from the fact that every vertex in every round of the DAG has at least $2f+1$ strong edges to vertices in the previous round.
In particular, it is easy to show by a counting argument that there is one vertex $u \in DAG_i[r_1]$ such that $f+1$ vertices in $DAG_i[r+2]$ has a strong edge to $u$.
Therefore, by quorum intersection, every vertex in $DAG_i[r+3]$ has a strong path to $u$.
Let $C \subseteq DAG_i[r]$, $|C| = 2f+1$ be the set of vertices that $u$ has a strong path to, then every vertex in $DAG_i[r+3]$ has a strong path to every vertex in $C$.
The lemma follows since there are at least $2f+1$ vertices in $DAG_i[r+3]$.

\end{proof}

\vspace{2mm}
Next, we use the fact that fallback leaders are hidden from adversary until the last round of a wave to prove the following:

\begin{claim}
\label{claim:probcommit}

Consider a party $p_i$ and a wave $w$ such that $p_i$ determines the vote type of all parties that reach $w$ in $DAG_i$ as fallback.
Then the probability of $p_i$ to commit the fallback vertex leader of $w$ is at least $2/3$.

\end{claim}

\begin{proof}
Let $r$ be the first round of $w$.
By the assumption, the vote type of all parties with vertices in $DAG_i[r+3]$ is fallback.
Therefore, by Claim~\ref{claim:core}, there are at least $2f+1$ vertices in the first round of $w$ that satisfy the fallback  commit rule.
That is, there is a set $C$ of $2f+1$ parties such that if any of them is elected to be the fallback leader, then $p_i$ will commit it.
Since the fallback leader is elected with the randomness produced in round $r+3$, the set $C$ is determined before the adversary learns the leader.
Therefore, even though the adversary fully controls delivery times, the probability for the elected leader to be in $C$ is at least $2f+1/3f+1 > 2/3$.

\end{proof}

\vspace{2mm}

\begin{claim}
\label{claim:somecommit}
For every wave $w$, there is an honest party that with probability $1$ commits a leader in a wave higher than $w$.
\end{claim}

\begin{proof}
Assume by a way of contradiction no honest party commits a leader in a wave higher than $w$.
By Observation~\ref{obs:DAGgrows}, for every round $r$ and honest party $p_i$, $DAG_i[r]$ eventually contains at least $2f+1$ vertices.
Moreover, by Claim~\ref{claim:honestfallback}, there is an honest party $p_i$ that determines the vote type of all parties that reach $w' > w +1$ in $DAG_i$ as fallback.
Therefore, by Claim~\ref{claim:probcommit}, the probability of $p_i$ to
commit the fallback leader in any wave $w' > w+1$ is at least $\frac{2}{3}$.
Hence, with probability $1$, there is a wave higher than $w$ that $p_i$ commits.

\end{proof}


\vspace{2mm}
We next use Claim~\ref{claim:somecommit} to prove Validity.

\begin{lemma}
\label{lem:validity}

Algorithms~\ref{alg:dataStructures}, \ref{alg:DAG}, \ref{alg:modes}, and~\ref{alg:commit} satisfy Validity.

\end{lemma}

\begin{proof}
Let $p_i$ be an honest party that calls $a\_bcast(b,r)$, we need to show that all honest parties output $a\_deliver(b,r,p_i)$ with probability $1$.
By the code $p_i$ pushes $b$ in the \emph{blockToPropose} queue.
By Observation~\ref{obs:DAGgrows}, $p_i$ advanced unbounded number of rounds and thus creates unbounded number of vertices.
Therefore, eventually $p_i$ will create a vertex $v_i$ with $b$ and reliably broadcast it.
By the Validity property of reliably broadcast, all honest parties will eventually add it to their DAG. 
That is, for every honest party $p_j$, there is a round number $r_i$ such that $v_i \in DAG_j[r_i]$. 
By the code of create\_new\_vertex, every vertex that $p_j$ creates after $v_i$ is added to $DAG_j[r_i]$ have a path to $v_i$ (either with strong links or weak links).

Therefor, by Claim~\ref{claim:somecommit}, there is an honest party $p_j$ that with probability $1$ commits a leader vertex with a path to $v_i$. 
Thus, by the code of order\_vertices, $p_i$ outputs $a\_deliver(b,r,p_i)$ with probability $1$.
Since $p_i$ is honest, we get that by Lemma~\ref{lem:agreement} (Agreement), all honest parties output $a\_deliver(b,r,p_i)$ with probability $1$.

\end{proof}

\paragraph{Integrity.}
\begin{lemma}
\label{lem:integrity}
Algorithms~\ref{alg:dataStructures}, \ref{alg:DAG}, \ref{alg:modes}, and~\ref{alg:commit} satisfy Integrity.
\end{lemma}

\begin{proof}
An honest party $p_i$ outputs $a\_deliver(v'.block,v'.round,\\ v'.source)$ only if node $v'$ is in $p_i$'s DAG (i.e., $v' \in \bigcup_{r > 0} DAG_i[r]$).
Node $v'$ is added to $p_i$'s DAG upon the reliable broadcast\\ $r\_deliver(v',v'.round, v'.source)$ event.
Therefore, the Lemma follows from the Integrity property of reliable broadcast.  
\end{proof}

\subsection{Partially Synchronous \sys}
\label{sub:ESproof}

The proof of the Integrity property is identical to the proof of Lemma~\ref{lem:integrity}.
For the rest of the properties, due to similarities between the protocols and to avoid argument duplication, we will follow the structure of Section~\ref{sec:fallback} and sometimes explain how to adapt claims' proofs.

To be consistent with the \sys with fallback presentation, waves here are also consist of 4 rounds, each with a pre-defined leader in the first and fourth rounds (we could have waves of 2 rounds since we do not have the fallback leader).

\paragraph{Total order.}
Observation~\ref{obs:relaible} applies in this case as well because the protocol to build the DAG is the same.
Claim~\ref{claim:votetype} trivially holds here since there is only one possible vote type and Claim~\ref{claim:samewave} holds since there are no fallback leaders.
The proofs of Claims~\ref{claim:commitmin} and~\ref{claim:commitinterval} apply to Algorithm~\ref{alg:ESBullshark} as well.
Therefore, Corollary~\ref{col:order} applies and since we use the same order\_vertices procedure in both protocols we get:

\begin{lemma}
\label{lem:ESorder}

Algorithms~\ref{alg:dataStructures}, \ref{alg:DAG}, \ref{alg:modes}, and~\ref{alg:ESBullshark} satisfy Total order.

\end{lemma}

\paragraph{Agreement and Validity.}

The proof of the Agreement property is identical to the proof of Lemma~\ref{lem:agreement} and Observation~\ref{obs:DAGgrows} holds since the algorithm to build the DAG is the same as in \sys with fallback.
To proof Validity for the eventually synchronous variant of \sys we do not need Claims~\ref{claim:honestfallback} and~\ref{claim:somecommit}.
Instead, we use the fact that GST eventually occurs.
We prove the protocol under the assumption that honest parties set their timeouts to be larger than $3\Delta$ and the following holds for the reliable broadcast building block:

\begin{property}
\label{prop:relaible}

Let $t$ be a time after GST. 
If an honest party reliably broadcasts a message at time $t$ or an honest party delivers a message at time $t$ , then all honest parties deliver it by time $t + \Delta$.

\end{property}

The above property is the equivalent to the reliable broadcast Validity and Agreement properties in the asynchronous model.
To the best of our knowledge, it is satisfied by all reliable broadcast protocol since before delivering a message honest parties echo it to all other honest parties.

\begin{claim}
\label{claim:EScommit}

Let $w$ be a wave such that all honest parties advances to the first round of $w$ after GST.
Let $p_1$ and $p_2$ be their first and second pre-defined leaders of $w$, respectively.
If $p_1$ and $p_2$ are honest, then all honest parties commit a leader in $w$.

\end{claim}

\begin{proof}
let $r$ be the first round of $w$.
First we show that all honest parties advance to round $r+1$ within $2\Delta$ time of each other.
By Observation~\ref{obs:DAGgrows}, all honest parties eventually advance to round $r+1$.
Let party $p_i$ be the first honest party that advances to round $r+1$ and denote by $t$ the time it happened. 
By the code of try\_advance\_round, $|DAG_i[r]| \geq 2f+1$.
By Property~\ref{prop:relaible}, by time $t + \Delta$ $|DAG_j[r]| \geq 2f+1$ for all honest parties.
Therefore, by Line~\ref{line:time2}, all honest party advance to round $r$ by time $t + \Delta$.
In particular, the first leader of wave $w$, $p_1$.
Thus, $p_1$ broadcasts its vertex $v_1$ in round $r$ no later than time $t + \Delta$, and by Property~\ref{prop:relaible}, all honest deliver it by time $t + 2\Delta$.
Therefore, by Line~\ref{line:advance1} and the code of try\_advance\_round, all honest parties advance to round $r+2$ by time $t + 2\Delta$.

Next we show that all honest parties advance to round $r+2$ with $3\Delta$ time of each other.
Since all honest parties advance to round $r+1$ within $2\Delta$ time of each other, then they start their timeouts at round $r+1$ within $2\Delta$ time of each other.
Let party $p_j$ be the first honest party that advances to round $r+2$.
If the first honest party waits for timeout (the if in Line~\ref{line:advance2}) to advance to round $r+2$, then all honest parties advance to round $r+2$ within $2\Delta$.
Otherwise, $p_j$ has $2f+1$ vertices in $DAG_j[r+1]$ with strong path to $v'$.
By property~\ref{prop:relaible}, all other honest parties will deliver this vertices and advance to round $r+2$ within $3\Delta$ from $p_j$.

By the assumption, the second leader of the wave, $p_2$, is honest and will broadcast vertex $v_2$ in round $r+2$ at most $3\Delta$ after the first honest party advances to $r+2$.
Since the timeouts are larger than $4\Delta$, all honest will advance to round $r+3$ within $\Delta$ of each other (by Line~\ref{line:advance3}, all honest wait to deliver the leader's vertex or for a timeout).
Moreover, they will all add a strong edge to $v_2$ in their vertex in round $r+3$.

Since all honest advance to round $r+3$ within $\Delta$ of each other and the timeouts are larger than $2\Delta$, they will all wait for each other's vertices before advancing to the next round.
Therefore, all honest will get $2f+1$ vertices in round $r+3$ with strong paths to the second vertex leader of the wave $v_2$.
Thus, all honest commit a leader in wave $w$. 

\end{proof}

\vspace{2mm}
The Validity property is proved under the assumption that eventually (after GST) there will be a wave in which both leaders are honest. For example, this assumption holds for every full permutation of the parties or if we maintain a fixed leader for the full wave.
To avoid repetition, we omit the proof of the following lemma as it is similar to the proof of Lemma~\ref{lem:validity}.
All we need to do to adapt it is to remove all appearances of "with probability $1$" and replace the reference to Claim~\ref{claim:somecommit} with Claim~\ref{claim:EScommit}.

\begin{lemma}
\label{lem:ECvalidity}

Algorithms~\ref{alg:dataStructures}, \ref{alg:DAG}, \ref{alg:modes}, and~\ref{alg:ESBullshark} satisfy Validity.

\end{lemma}

\end{document}